\newif\ifShowJournalIdeas
\ShowJournalIdeastrue

\documentclass{article}
\usepackage[hmarginratio=3:2]{geometry}

\usepackage{tikz}
\usetikzlibrary{matrix,fit,decorations.pathreplacing}

\pgfdeclarelayer{background}
\pgfdeclarelayer{foreground}
\pgfsetlayers{background,main,foreground}

\newcommand{\headerL}[2]{\node[rectangle, rounded corners, draw=black, thick, inner sep=4pt, fill=black!10,xshift=12pt,yshift=-1pt,anchor=south west, minimum height=4.2ex] at (#1.north west) {#2};}
\newcommand{\headerR}[2]{\node[rectangle, rounded corners, draw=black, thick, inner sep=3pt, fill=black!02,xshift=-12pt,yshift=-1pt,anchor=south east] at (#1.north east) {#2};}

\newcommand{\paramBox}[5]{
		\node[inner sep = 2pt,anchor = south west] (A#1) at (#2) {
		\begin{tabularx}{\TableWidth}{s|m|b}
			& & \\ [-10.5pt] 
			#4
		\end{tabularx}
		};
		\node[inner sep=-0.6pt,draw=white,ultra thick,rounded corners,fit=(A#1)] () {};
		\node[inner sep=-1.4pt,draw=white,ultra thick,rounded corners,fit=(A#1)] () {};
		\node[inner sep=-2pt,draw=black,thick,rounded corners,fit=(A#1)] (#1) {};
		\begin{pgfonlayer}{foreground}
			\headerL{#1}{#3}
			\headerR{#1}{#5}
		\end{pgfonlayer}
}

\usepackage[utf8]{inputenc} 
\usepackage{amsmath}
\usepackage{hyperref}
\usepackage{algorithm}
\usepackage{amssymb}
\usepackage{nicefrac}
\usepackage{todonotes}
\usepackage{dsfont}
\usepackage{aligned-overset}

\usepackage{tabularx}
\usepackage{colortbl}
\usepackage{booktabs}
\usepackage{multirow}

\usepackage[nameinlink,capitalize]{cleveref}
\usepackage[]{hyperref}

\usepackage{etoolbox}

\newcommand{\appsymb}{$\star$}
\newcommand{\appref}[1]{{\hyperref[proof:#1]{\appsymb}}}

\usepackage{amsthm}
\theoremstyle{definition}
\newtheorem{definition}{Definition}[section]
\theoremstyle{plain}
\newtheorem{theorem}{Theorem}[section]
\newtheorem{proposition}{Proposition}[section]
\newtheorem{observation}{Observation}[section]

\newtheorem{lemma}{Lemma}[section]
\newtheorem{corollary}{Corollary}[section]
\newtheorem{rrule}{Reduction Rule}[section]

\newtheorem{claim}{Claim}{\upshape\itshape}{\upshape\rmfamily}
\newenvironment{claimproof}{{\noindent\textit{Proof of Claim. }}}{\hfill$\blacksquare$}

\crefname{rrule}{Reduction Rule}{Reduction Rules}
\Crefname{theorem}{Theorem}{Theorems}
\crefname{theorem}{Thm.}{Thms.}
\Crefname{proposition}{Proposition}{Propositions}
\crefname{proposition}{Prop.}{Props.}
\Crefname{observation}{Observation}{Observations}
\crefname{observation}{Obs.}{Obs.}
\crefname{definition}{Definition}{Definitions.}
\crefname{definition}{Def.}{Defs.}
\crefname{corollary}{Cor.}{Cors.}
\Crefname{corollary}{Corollary}{Corollaries}
\Crefname{claim}{Claim}{Claims}
\Crefname{figure}{Figure}{Figures}

\newcommand{\Oh}{\ensuremath{\mathcal{O}}}

\newcommand{\problemdef}[3]{
  \vspace{.5ex}
  \begin{quote}
      \normalsize\textsc{#1} \smallskip \\
      \begin{tabularx}{0.9\textwidth}{@{}l@{\hspace{3pt}}X}
        \normalsize\textbf{Input:}    & \normalsize#2 \\
        \normalsize\textbf{Question:} & \normalsize#3
      \end{tabularx}
  \end{quote}
  \vspace{.5ex}
}

\newcommand{\PHC}{{\normalfont coNP $\subseteq$ NP/poly}} 

\newcommand{\probFull}[1]{{\normalfont\textsc{#1$\alpha$-Fixed Cardinality Graph Partitioning}}}
\newcommand{\probShort}[2]{{\normalfont\textsc{#2#1$\alpha$-FCGP}}}

\newcommand{\borderFocused}{Degrading}
\newcommand{\internalFocused}{Non-Degrading}

\newcommand{\borderFocusedS}{degrading}
\newcommand{\internalFocusedS}{non-degrading}

\newcommand{\prob}{\probShort{}}
\newcommand{\probAno}{\probShort{Annotated }} 
\newcommand{\probBor}{\probShort{\borderFocused{} }} 
 
\newcommand{\probAnoBor}{\probShort{Annotated \borderFocused{} }}

\newcommand{\probMax}{\probShort{Max }}
\newcommand{\probBorMax}{\probShort{\borderFocused{} Max }} 
\newcommand{\probIntMax}{\probShort{\internalFocused{} Max }} 
\newcommand{\probAnoMax}{\probShort{Annotated Max }}
\newcommand{\probAnoBorMax}{\probShort{Annotated \borderFocused{} Max }}

\newcommand{\probMin}{\probShort{Min }}
\newcommand{\probAnoMin}{\probShort{Annotated Min }} 
\newcommand{\probBorMin}{\probShort{\borderFocused{} Min }} 
\newcommand{\probIntMin}{\probShort{\internalFocused{} Min }} 
\newcommand{\probAnoBorMin}{\probShort{Annotated \borderFocused{} Min }}

\newcommand{\contribution}{\mathrm{cont}} 
\DeclareMathOperator{\val}{val} 

\newcommand{\counter}{\mathsf{counter}}

\DeclareMathOperator{\degCounter}{deg^{+c}}

\begin{document}

\title{Covering Many (or Few) Edges with $k$ Vertices in~Sparse~Graphs\footnote{A preliminary version of this article appeared in Proceedings of the 39th International Symposium on Theoretical Aspects of Computer Science (STACS~'22), volume 219 of LIPIcs, pages 42:1--42:18.}}

\author{Tomohiro Koana$^1$\footnote{Supported by the Deutsche Forschungsgemeinschaft (DFG), project DiPa, NI 369/21.}
 \and Christian Komusiewicz$^2$
 \and André Nichterlein$^1$
 \and Frank Sommer$^2$\footnote{Supported by the Deutsche Forschungsgemeinschaft (DFG), project EAGR, {KO~3669/{6-1}}.} 
 }
\date{%
    $^1$ Algorithmics and Computational Complexity, Technische Universität Berlin, Germany\\%
    $^2$ Fachbereich Mathematik und Informatik, Philipps-Universität Marburg, Germany\\[2ex]%
}






\maketitle






\begin{abstract}
	We study the following fixed-cardinality optimization problem in a maximization and a minimization variant.
	For a fixed~$\alpha$ between zero and one we are given a graph and two numbers~$k \in \mathds{N}$ and~$t \in \mathds{Q}$.
	The task is to find a vertex subset~$S$ of exactly~$k$ vertices that has value at least~$t$ in the maximization variant or at most~$t$ in the minimization variant.
	Here, the value of a vertex set computes as $\alpha$ times the number of edges with exactly one endpoint in~$S$ plus~$1-\alpha$ times the number of edges with both endpoints in~$S$.
	These two problems generalize many prominent graph problems, such as \textsc{Densest $k$-Subgraph}, \textsc{Sparsest $k$-Subgraph}, \textsc{Partial Vertex Cover}, and~\textsc{Max ($k$,$n-k$)-Cut}.
	
	In this work, we complete the picture of their parameterized complexity on several types of sparse graphs that are described by structural parameters. 
	In particular, we provide kernelization algorithms and kernel lower bounds for these problems. 
	A somewhat surprising consequence of our kernelizations is that \textsc{Partial Vertex Cover} and \textsc{Max $(k,n-k)$-Cut} not only behave in the same way but that the kernels for both problems can be obtained by the same algorithms.
\end{abstract}

\maketitle

\section{Introduction}
Fixed-cardinality optimization problems are a well-studied class of
graph problems where one seeks, for a given graph~$G$, a vertex set~$S$
of size~$k$ such that~$S$ optimizes some objective
function~$\val_G(S)$~\cite{BEHM06,CaiCC06,Cai08,KS15}. Prominent
examples of these problems are \textsc{Densest
  $k$-Subgraph}~\cite{BGLMP17,Fei97,KS15}, \textsc{Sparsest
  $k$-Subgraph}~\cite{HKYT17,HMTYS15,WBG16}, \textsc{Partial Vertex
  Cover}~\cite{AFS11,GNW07,KLR08}, and~\textsc{Max ($k$,$n-k$)-Cut}~\cite{Cai08,SZ18,SZ17}.

A common thread in these example problems is that they are
formulated in terms of the number of edges that have one or two
endpoints in~$S$: In the decision version of \textsc{Densest
  $k$-Subgraph} we require that there are at least~$t$ edges with both
endpoints in~$S$. Conversely, in \textsc{Sparsest $k$-Subgraph}
we require that at most~$t$ edges have both endpoints in~$S$. In
\textsc{Partial Vertex Cover} we require that at least~$t$ edges have
at least one endpoint in~$S$. Finally, in \textsc{Max ($k$,$n-k$)-Cut}
we require that at least~$t$ edges have exactly one endpoint in~$S$.

We study the following related problem first defined by Bonnet et al.~\cite{BEPT15}.\footnote{On the face of it, the definition of Bonnet et al.~\cite{BEPT15} seems to be more general as it has separate weight parameters for the internal and outgoing edges. It can be reduced to our formulation by adapting the value of~$t$ and thus our results also hold for this formulation.}

\problemdef
{\probFull{Max } (\probMax{})}
{A graph $G$, $k \in \mathds{N}$, and~$t \in \mathds{Q}$.}
{Is there a set $S$ of exactly $k$ vertices such that \newline
\begin{minipage}{25em}\[\val(S)\coloneqq (1-\alpha) \cdot m(S) + \alpha \cdot m(S, V(G) \setminus S)\ge t\ ?\]\end{minipage}}
\vspace{1.2ex}

Here, $\alpha\in [0,1]$ and~$m(S)$ denotes the number of edges with two endpoints in~$S$ and $m(S, V(G) \setminus S)$ denotes the number of edges with exactly one endpoint in~$S$.
Naturally, one may also consider the minimization problem, denoted as \probFull{Min } (\probMin{}), where we are looking for a set~$S$ such that~$\val(S)\le t$.

The value of~$\alpha$ describes how strongly edges with exactly one endpoint in~$S$ influence the value of~$S$ relative to edges with two endpoints in~$S$. 
For~$\alpha = 1 / 3$, edges with two endpoints in~$S$ count twice as much as edges with one endpoint in~$S$ and, thus, every vertex contributes exactly its degree to the value of~$S$. Hence, in this case, we simply want to find a vertex set with a largest or smallest degree sum. 

More importantly, \probMax{} and \probMin{} contain all of the above-mentioned problems as special cases (see \cref{fig:problem-def}). 
\begin{figure}[t]
	\centering
	\begin{tikzpicture}[xscale=1.25,node distance = 12pt]
		\newcommand{\probh}{0.75}

		\tikzset{prob/.style={rectangle, rounded corners, draw=black, thick, inner sep=3pt}}
		
		\node at  (-0.5,0)   {$\alpha$};
		
		\node[prob,minimum width=3.75cm] (DkS) at  (0.75, \probh - 0.1)   {\small \textsc{Densest~$k$-Subgraph}};
		\node[prob,minimum width=3.75cm] (SkS) at  (0.75,-\probh - 0.1)   {\small \textsc{Sparsest~$k$-Subgraph}};

		\node[prob,right=of DkS,minimum width=2.5cm] (MaDS) {\small \textsc{Max Deg Sum}};
		\node[prob,right=of SkS,minimum width=2.5cm] (MiDS) {\small \textsc{Min Deg Sum}};

		\node[prob,right=of MaDS,minimum width=3.1cm] (MaPV) {\small \textsc{Max PVC}};
		\node[prob,right=of MiDS,minimum width=3.1cm] (MiPV) {\small \textsc{Min PVC}};

		\node[prob,right=of MaPV,minimum width=3.1cm] (MaC) {\small \textsc{Max ($k$,$n-k$)-Cut}};
		\node[prob,right=of MiPV,minimum width=3.1cm] (MiC) {\small \textsc{Min ($k$,$n-k$)-Cut}};

		\draw[black!50] (0, 0.1) -- (DkS);
		\draw[black!50] (0,-0.1) -- (SkS);
		\draw[black!50] (3.33, 0.1) -- (MaDS);
		\draw[black!50] (3.33,-0.1) -- (MiDS);
		\draw[black!50] (5, 0.1) -- (MaPV);
		\draw[black!50] (5,-0.1) -- (MiPV);
		\draw[black!50] (10, 0.1) -- (MaC);
		\draw[black!50] (10,-0.1) -- (MiC);

		\draw[thick] (0,0) -- (10,0);

		\draw[thick] (0,-0.1) -- (0,0.1);
		\node at  (0,-0.35)   {0};

		\draw[thick] (3.33,-0.1) -- (3.33,0.1);
		\node at  (3.2,-0.35)   {\nicefrac{1}{3}};

		\draw[thick] (5,-0.1) -- (5,0.1);
		\node at  (5,-0.35)   {\nicefrac{1}{2}};

		\draw[thick] (10,-0.1) -- (10,0.1);
		\node at  (10,-0.35)   {1};

	\end{tikzpicture}
	\caption{Problem definition cheat sheet.}
	\label{fig:problem-def}
\end{figure}
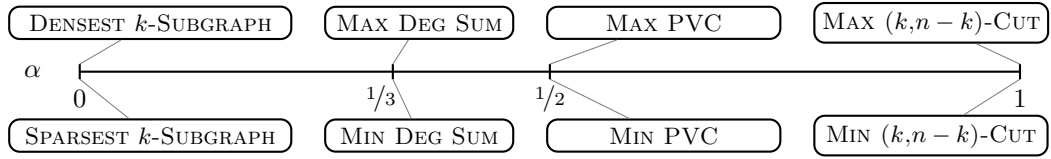
For~$\alpha=0$,  only the edges with both endpoints in~$S$ count and thus \textsc{Densest $k$-Subgraph} is \probMax{} and \textsc{Sparsest $k$-Subgraph} is \probMin{} in this case. For~$\alpha=1$, only edges with exactly one endpoint in~$S$ count and \textsc{Max $(k, n - k)$-Cut} is \probMax{} in this case.
 Finally, \textsc{Partial Vertex Cover} (\textsc{MaxPVC}) is \probMax{} with~$\alpha = 1 / 2$ as all edges with at least one endpoint in~$S$ count the same.  
Consequently, there exist values of~$\alpha$ such that \probMax{} and \probMin{} are NP-hard and W[1]-hard on general graphs with respect to the natural parameter~$k$~\cite{Cai08,CFK+15,DF13,GNW07}. 
This hardness makes it interesting to study these problems on input graphs with special structure. For example, Bonnet et al.~\cite{BEPT15} and~Shachnai and Zehavi~\cite{SZ17} studied this problem on bounded-degree graphs.

We continue this line of research and give a complete picture of
the parameterized complexity of \probMin{} and \probMax{} on several
types of sparse graphs that are described by structural parameters. In
particular, we provide kernelization algorithms and kernel lower
bounds for these problems, see \cref{fig:overview} for an overview.
\begin{figure}[t]

	\definecolor{ColorOpen}{RGB}{255, 255, 255}
	\definecolor{ColorHard}{RGB}{255, 100, 100} 
	\definecolor{ColorFPT}{RGB}{255, 200, 100}
	\definecolor{ColorPPK}{RGB}{255, 255, 100}
	\definecolor{ColorPK}{RGB}{100, 255, 50}   

	\newcolumntype{C}[1]{>{\centering\let\newline\\\arraybackslash\hspace{0pt}}p{#1}}
	\newcolumntype{D}{>{\centering\let\newline\\\arraybackslash\hspace{0pt}}X}

	\newcolumntype{b}{>{\hsize=1.3\hsize}D}
	\newcolumntype{m}{>{\hsize=1\hsize}D}
	\newcolumntype{s}{>{\hsize=.7\hsize}D}

	\newcommand{\columnSep}{0.5cm}
	\newcommand{\rowDist}{1.6cm}

	\newcommand{\TableWidth}{6.7cm}
	\centering \small
	\begin{tikzpicture}

		\paramBox{vc}{0,0}{Vertex Cover Number~$\mathsf{vc}$}{
				 \cite{BJK14} & \cellcolor{ColorPK} \cref{prop:vertex-cover:poly-kernel-max} & \cellcolor{ColorPK}
				\\\hline
				\multicolumn{2}{c|}{\cellcolor{ColorPK}} & \cellcolor{ColorPK} \cref{prop:vertex-cover:poly-kernel-min}\\
		}{\cref{sec:vc}}
		
		\paramBox{max}{\TableWidth + \columnSep,0}{Max Degree~$\Delta$}{
				\multicolumn{2}{c|}{\cellcolor{ColorFPT} \cref{thm:maxdeg:no-poly-kernel-const-delta}}
				& \cellcolor{ColorPK} \\\hline
				
				\multicolumn{2}{c|}{\cellcolor{ColorPK} }
				& \cellcolor{ColorFPT}\cref{thm:maxdeg:no-poly-kernel-const-delta} \\
		}{\cref{sec:maxdeg}}

		\paramBox{hind}{0,-\rowDist}{$h$-index}{
				\multicolumn{2}{c|}{\cellcolor{ColorFPT} \cref{prop:hindex:fpt-internal-max}} 
				& \cellcolor{ColorPK}\cref{prop:h-index:poly-kernel} \\ \hline
				\multicolumn{2}{c|}{\cellcolor{ColorPK}} & \cellcolor{ColorFPT} \\
		}{\cref{sec:vc}}

		\paramBox{dg}{0,-2*\rowDist}{degeneracy~$d$}{
				\multicolumn{2}{c|}{\cellcolor{ColorHard}\cref{thm-hardness-small-alpha-for-deg-and-closure}} 
				& \cellcolor{ColorPPK} \cref{thm:degeneracy:max-border-case} \\ \hline
				\multicolumn{2}{c|}{\cellcolor{ColorPK}\cref{thm-minimization-kernel-d+k}} & \cellcolor{ColorFPT}\cref{lem-minimization-fpt-for-d} \\
		}{\cref{sec:degeneracy}}

		\paramBox{cc}{\TableWidth + \columnSep,-1.25*\rowDist}{$c$-closure}{
				\multicolumn{2}{c|}{\cellcolor{ColorHard}\cref{thm-hardness-small-alpha-for-deg-and-closure}} 
				& \cellcolor{ColorPPK} \cref{thm:degrading:closure} \\ \hline
				\cellcolor{ColorPK} \cite{KKS20a} & \cellcolor{ColorPPK} \cref{thm:degrading:closure} & \cellcolor{ColorHard}\cref{thm:nondeg:max:c} \\
		}{\cref{sec:closure}}
		
		\draw[thick] (vc) -- (hind) -- (dg);
		\draw[thick] (max) -- (hind);
		\draw[thick] (max) -- (cc);

		\paramBox{legend}{1.1,-3*\rowDist}{Parameter~$p$}{
				$\alpha = 0$  & $\alpha \in (0,1/3)$ & $\alpha \in (1/3,1]$ \\ \hline $\alpha = 0$  & $\alpha \in (0,1/3)$ & $\alpha \in (1/3,1]$ \\ 
		}{Section with results}
		\node[anchor=south west] at (0,-3*\rowDist - 1pt) {
			\begin{tabularx}{1.05cm}{l}
				Max: \\ \hline
				Min:
			\end{tabularx}
		};

		\node[draw=black,very thick,inner sep = 0pt,anchor=south east] at (\textwidth,-3*\rowDist + 2pt) {
			\begin{tabularx}{5.5cm}{D}
				\cellcolor{ColorHard} W[1]-hard wrt.~$k$ for constant~$p$\\ \hline 
				\cellcolor{ColorFPT} FPT, no $k^{g(p)}$ kernel for any~$g$ \\ \hline
				\cellcolor{ColorPPK} $k^{\Oh(p)}$ kernel, no~$k^{o(p)}$ kernel \\ \hline
				\cellcolor{ColorPK} $(k+p)^{\Oh(1)}$ kernel \\ \hline
				\cellcolor{white} $2^{p}$ kernel, no $(k+p)^{\Oh(1)}$ kernel \\ \hline
			\end{tabularx}
		};
		
	\end{tikzpicture}

	\caption{
		Overview of our results. 
		Each box displays the parameterized results (see also bottom right) with respect to~$k$ and the corresponding parameter~$p$ for all variants (maximization, minimization, and all $\alpha \in [0,1]$, see bottom left).
		Note that the split of the boxes is not proportional to the corresponding values of~$\alpha$.
		See \Cref{sec:prelim} (paragraph ``Graph parameter definitions.'') for the definitions of the parameters.
		A line from a box for parameter~$p$ to a box \emph{above} for parameter~$p'$ implies that~$p \in \Oh(p')$ on all graphs.
		Thus, hardness results propagate downwards along lines and tractability results propagate upwards along lines.}
	\label{fig:overview}
\end{figure}
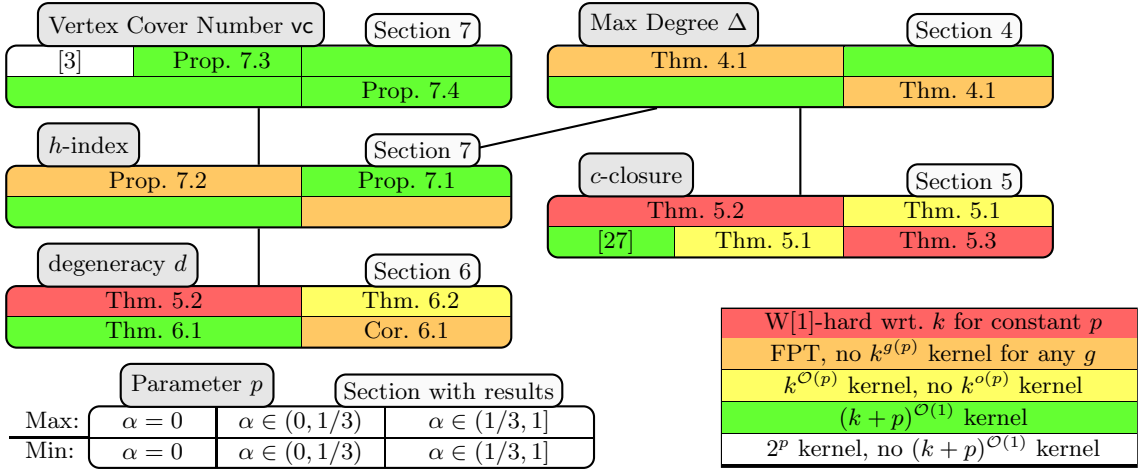



\paragraph{Known results.}

\textsc{MaxPVC} can be solved in  $\Oh^*((\Delta + 1)^k)$~time\footnote{The~$\Oh^*$ notation hides polynomial factors in the input size.} where~$\Delta$ is the maximum degree of the input graph~\cite{RS08}. 
For the degeneracy $d$, Amini et al.~\cite{AFS11} gave an $\Oh^*((dk)^k)$-time algorithm which was recently improved to an algorithm with running time~$\Oh^*(2^{\Oh(dk)})$~\cite{PY22}.
Bonnet et al.~\cite{BEPT15} showed that  in~$\Oh^*(\Delta^k)$~time one can solve \probMax{}  for all~$\alpha > 1/3$ and \probMin{} for all~$\alpha < 1/3$. Bonnet et al.~\cite{BEPT15} call these two problem cases \emph{\borderFocusedS{}}.
This name reflects the fact that in \probMax{} with~$\alpha > 1/3$, adding a vertex~$v$ to a set~$S$ increases the value at least as much as adding~$v$ to a superset of~$S$.\footnote{Note that this matches the definition of submodularity.}     
This is because here one edge with both endpoints in~$S$ is less valuable than two edges each with one endpoint in~$S$. In~\probMin{} this effect is reversed since we aim to minimize~$\val$. The other problem cases are called \emph{\internalFocusedS{}}.
For~\internalFocusedS{} problems, Bonnet et al.~\cite{BEPT15} achieved a running time of~$\Oh^*((\Delta k)^{\Oh(k)})$ and asked whether they can also be solved in~$\Oh^*(\Delta^{\Oh(k)})$~time. This question was answered positively by Shachnai and Zehavi~\cite{SZ17}, who showed that \probMax{} and \probMin{} can be solved in~$\Oh^*(4^{k+o(k)}\Delta^k)$ time.

Kernelization has been studied only for special cases.
\textsc{Max ($k$,$n-k$)-Cut} admits a polynomial problem kernel when parameterized by~$t$~\cite{SZ18}. This also gives a polynomial kernel for~$k+\Delta$ since instances with~$t>\Delta k$ are trivial no-instances.
It is also known that \textsc{Sparsest~$k$-Subgraph} admits a kernel with $\gamma k^2$~vertices~\cite{KKS20a}.
Here,~$\gamma$ is a parameter bounded by~$\max(c, d+1)$ where~$c$ is the so-called~$c$-closure of the input graph~\cite{FRSWW20}. We will describe this parameter in more detail below.
In contrast, \textsc{Densest-$k$ Subgraph} is unlikely to admit a polynomial problem kernel when parameterized by~$\Delta + k$ since \textsc{Clique} is a special case.

Independent of our work, a polynomial \emph{compression} for \textsc{MaxPVC} (the special case of \probMax{} with~$\alpha=1/2$) of size $(dk)^{\Oh(d)}$ was recently discovered by Panolan and Yaghoubizade~\cite{PY22}.
This result and our kernel of size~$k^{\Oh(d)}$ (\Cref{thm:degeneracy:max-border-case}) both independently answer an open question of Amini, Fomin, and Saurabh~\cite{AFS11}.
They asked whether \textsc{MaxPVC} admits a polynomial kernel in planar graphs.
Our results directly imply such a kernel since planar graphs have degeneracy at most~$5$. 


More broadly, for many graph problems which are W[1]-hard with respect to the solution size parameter~$k$, the study of kernelization on sparse graphs~\cite{CGH17,CPPW12,KKS21,PRS12} or on graphs with bounded~$c$-closure~\cite{KMRS22,KKS20,KKS20a,KKS21} has received a lot of attention in the recent years. 
In particular,~Kanesh et al.~\cite{KMRS22} posed as an open question to determine the complexity of \textsc{Partial Vertex Cover} with respect to~$k$ on~$c$-closed graphs; we answer this question positively by providing a kernel of size~$k^{\Oh(c)}$ (see \Cref{thm:closure:kernel}). 

\paragraph{Our results.}

We provide a complete picture of the
parameterized complexity of \probMax{} and \probMin{} for all~$\alpha$
with respect to the combination of~$k$ and five parameters describing
the graph structure: the maximum degree~$\Delta$ of~$G$, the $h$-index
of~$G$, the degeneracy of~$G$, the $c$-closure of~$G$, and the vertex cover number~$\mathsf{vc}$ of~$G$.
With the exception of the $c$-closure, all parameters are sparseness
measures. The $c$-closure, first described by Fox et
al.~\cite{FRSWW20}, measures how strongly a graph adheres to the
triadic closure principle. Informally, the closure of a graph is small
whenever all vertices with many common neighbors are also neighbors of
each other. For a formal definition of all parameters refer to \Cref{sec:prelim}.

Our results are summarized by \cref{fig:overview}. On a very
general level, our main finding suggests that the \borderFocusedS{} problems
are much more amenable to FPT algorithms and kernelizations than their
\internalFocusedS{} counterparts. No such difference is observed
when considering the running time of FPT algorithms for the parameter~$k+\Delta$ but it
becomes striking in the context of kernelization and when using
 secondary parameters that are smaller than~$\Delta$.  Given the importance of the distinction
between the \borderFocusedS{} and \internalFocusedS{} cases, we
distinguish these subcases of \probMax{} and \probMin{} by name (\probBorMax{}, \probIntMax{}, \probBorMin{}, \probIntMin{}). 
We use the term \prob{} when we simultaneously refer to \probMax{} and \probMin{}.

On a technical level, a first contribution is the introduction of an annotated version of the problem that keeps track of removed vertices by using vertex weights (called $\counter$). This simplifies the technical details of dealing with vertices identified as (not) being part of a solution. Moreover, this annotated version allows us to formulate the rules for dealing with such vertices in a unified manner for all special cases of~\prob{}. As one may expect, the general approach for the kernel is to reduce to the annotated problem, perform data reduction on the annotated instance, and then reduce back to the non-annotated problem. There is, however, one technical difficulty in the last step: The size of the non-annotated instance depends not only on the number of vertices of the annotated instance but also on the largest $\counter$-value. Now, for the parameter combination of~$k$ and~$\Delta$ this is not really a problem since the maximum counter value is bounded in~$\Oh(\Delta)$ for fixed~$\alpha$.
For the more sophisticated kernelizations for the degeneracy~$d$ and the $c$-closure we now perform two steps: First, we decrease the maximum degree of the instance as this allows us, in principle, to use the kernel for~$k+\Delta$. Second, we decrease the largest $\counter$-value as this value may not be bounded in~$\Oh(k+\Delta)$ after the first step which decreases the maximum degree of the instance.

Now, we describe our idea to decrease the maximum degree.
In this high-level description, we only focus on the arguments for maximization.
The minimization variant can be handled similar.
We make use of Ramsey bounds. 
More precisely, the Ramsey bounds help us to find a large independent set $I$ such that all vertices outside of~$I$ have only a bounded number of neighbors in $I$. 
This then allows to prove by pigeonholing the following for the vertex $v$ of $I$ whose addition currently gives the smallest increase to the objective function: 
No matter what the optimal solution selects outside of~$I$, there is always some vertex of~$I\setminus \{v\}$ that increases the objective function at least as much as~$v$. 
For the parameter~$c$, we also need an additional pigeonhole argument excluding large cliques in order to apply the Ramsey bound. 
For the parameter $d$, we establish a new constructive Ramsey bound for $K_{i,j}$-free graphs that may be of independent interest. 
To decrease the largest $\counter$-values, we describe reduction rules that make use of the following crucial observation:  
If there exists a vertex~$v$ whose addition increases the objective function by far more than~$t/k$, then we can add~$v$ to the solution and if its addition increases the objective function by far less than~$t/k$, then we can remove~$v$. 
Note that we can think of~$t/k$ as the increase of the objective function that is needed on average to find a solution of value at least~$t$.  

We remark that when we describe the kernel sizes for $\alpha > 0$ (for example, \Cref{prop:general:delta}), a factor of $\alpha^{-x}$, where~$x$ is a small constant, is hidden by the $\Oh$ notation.
We would like to emphasize, however, that the exponents in the kernel size such as $\Oh(c)$ and $\Oh(d)$ do not depend on $\alpha^{-1}$.
In contrast, the lower bounds such as \Cref{thm:maxdeg:no-poly-kernel-const-delta} hold indeed for all~$\alpha$ in the range corresponding to the case. 

We believe that this general approach could be useful for other parameterizations that are not considered in this work.  
A somewhat surprising consequence of our kernelizations is that \textsc{Partial Vertex Cover} and \textsc{Max $(k,n-k)$-Cut} not only behave in the same way but that the kernels for both problems can also be obtained by the same algorithms. 




\section{Preliminaries}
\label{sec:prelim}

For~$q \in \mathds{N}$, we write~$[q]$ to denote the set~$\{ 1, 2, \dots, q \}$.
For a graph~$G$, we denote its \emph{vertex set} by~$V(G)$ and its \emph{edge set} by~$E(G)$.
Furthermore, by~$n(G)\coloneqq |V(G)|$ we denote the \emph{number of vertices} and by~$m(G)\coloneqq |V(G)|$ we denote the \emph{number of edges} of~$G$.
Let~$X,Y\subseteq V(G)$ be vertex subsets.
We use~$G[X]$ to denote the \emph{subgraph induced} by~$X$.
Let $G - X$ denote the graph obtained by \emph{removing} the vertices in~$X$.
We denote by~$N_G(X)\coloneqq\{ y \in V(G) \setminus X \mid xy \in E(G), x \in X \}$ the \emph{open neighborhood} and by~$N_G[X]\coloneqq N_G(X) \cup X$ the \emph{closed neighborhood} of~$X$.
We also use the notation $N^{\cap}_G(X) \coloneq \bigcap_{x \in X}N(x)$.
By~$E_G(X,Y)\coloneqq\{xy\in E(G)\mid x\in X, y\in Y\}$ we denote the set of edges \emph{between}~$X$ and~$Y$.
As a shorthand, we set~$E_G(X)\coloneqq E_G(X,X)$.
Furthermore, we denote by~$m_G(X,Y)\coloneqq|E_G(X,Y)|$ and~$m_G(X)\coloneqq|E_G(X)|$ the \emph{sizes} of these edge sets.
For all these notations, when $X$ is a singleton~$\{ x \}$ we may write~$x$ instead of~$\{x\}$.
Let~$v \in V(G)$.
We denote the \emph{degree} of~$v$ by~$\deg_G(v)$.
We drop the subscript~$\cdot_G$ when it is clear from context.
We call~$v$ \emph{isolated} if~$\deg_G(v) = 0$ and \emph{non-isolated} otherwise.
We also say that~$v$ is a \emph{leaf vertex} if~$\deg_G(v) = 1$ and a \emph{non-leaf vertex} if~$\deg_G(v) \ge 2$.

\paragraph{Graph parameter definitions.}
For more information on parameterized complexity, including the definition of fixed-parameter tractability, kernelization, parameterized reductions, and W[1]-hardness, we refer to the standard monographs~\cite{CFK+15,DF13}.
A vertex cover of a graph is a set of vertices that covers all of its edges. We denote the size of a smallest vertex cover of a graph~$G$ by~$\mathsf{vc}_G$.
The maximum and minimum degree of~$G$ are~$\Delta_G \coloneqq \max_{v \in V(G)} \deg_G(v)$ and~$\delta_G \coloneqq \min_{v \in V(G)} \deg_G(v)$, respectively.
The degeneracy of~$G$ is~$d_G \coloneqq\max_{S \subseteq V(G)} \delta_{G[S]}$.
The~\emph{${h}$-index} of a graph~$G$ is the largest integer~$h_G$ such that~$G$ has at least~$h_G$ vertices of degree at least~$h_G$~\cite{ES12}.
We say that~$G$ is~$c$-closed for~$c\coloneqq \max(\{ 0 \} \cup \{ |N_G(u) \cap N_G(v)| \mid uv \notin E(G) \}) + 1$~\cite{FRSWW20}.

\paragraph{Ramsey numbers.}
Ramsey's theorem states that for every $p, q \in \mathds{N}$, there exists an integer~$R(p, q)$ such that any graph on at least $R(p, q)$ vertices contains either a clique of size~$p$ or an independent set of size $q$.
The numbers~$R(p, q)$ are referred to as \emph{Ramsey numbers}.
Although the precise values of Ramsey numbers are not known, some upper bounds have been proven.
For instance, it holds that $R(p, q) \le \binom{p + q - 2}{p - 1}$ (see e.g.\ \cite{Juk11}).
The proof for this upper bound is constructive.
More precisely, given a graph $G$ on at least $\binom{p + q - 2}{p - 1}$ vertices, we can find in time $n^{\Oh(1)}$ either a clique of size $p$ or an independent set of size $q$.

{
\paragraph{Weak compositions.}
Our kernel lower bounds are based on weak~$q$-compositions.



\begin{definition}[\cite{DM12,HW12}]
  Let~$q\ge 1$ be an integer, let~$L_1\subseteq \{0,1\}^*$ be a classic (non-parameterized) problem, and let~$L_2\subseteq \{0,1\}^*\times\mathds{N}$ be a parameterized  problem.  A \emph{weak~$q$-composition} from~$L_1$ to~$L_2$ is a polynomial time algorithm that on input~$x_1, \ldots ,x_{t^q}\in\{0,1\}^n$ outputs an instance~$(y,k')\in\{0,1\}^*\times\mathds{N}$ such that:
\begin{itemize}
    \item $(y,k')\in L_2 \Leftrightarrow x_i\in L_1 \text{ for some } i \in [t^q]$, and
    \item $k'\le t\cdot n^{\Oh(1)}$.
  \end{itemize}
\end{definition}

\begin{lemma}[\cite{CFK+15,DM12,HW12}] \label{lemma:is:noq}
  Let~$q\ge 1$ be an integer, let~$L_1\subseteq \{0,1\}^*$ be a classic NP-hard problem, and let~$L_2\subseteq \{0,1\}^*\times\mathds{N}$ be a parameterized problem.
  The existence of a weak~$q$-composition from~$L_1$ to~$L_2$ implies that~$L_2$ has no compression of size~$\Oh(k^{q-\epsilon})$ for any~$\epsilon >0$, unless \PHC.
\end{lemma}

}

\section{A Data Reduction Framework via Annotation}
\label{sec:annotation}

In this section, we introduce an annotated variant which gives more options for encoding information in the instances, allowing easier handling for kernelization.
Moreover, to avoid repeating certain basic arguments, we provide general data reduction rules and statements used in the subsequent sections.
Finally, we describe how to reduce from the annotated to the non-annotated problem variants in polynomial time.

In the annotated problem variant we encode that some vertices are decided to be in the solution and some vertices are decided to not be in the solution.
To this end, we have additionally as input a (possibly empty) partial solution $T \subseteq V(G)$.
Moreover, for each vertex we store a number~$\counter\colon V \rightarrow \mathds{N}$ which encodes the number of deleted neighbors not in the solution.
We will assume throughout the paper that $\counter(v) = 0$ for every $v \in T$.
For a set~$S \subseteq V(G)$, we set
\begin{itemize}
	\item $\counter(S) \coloneqq \sum_{v \in S} \counter(v)$ and
	\item $\val_G(S) \coloneqq \alpha (m(S, V(G) \setminus S) + \counter(S)) + (1 - \alpha) m(S)$.
\end{itemize}
For a vertex~$v \in S$, we set~$\degCounter(v) \coloneqq \deg(v) + \counter(v)$.

\problemdef
{\probAnoMax{}}
{A graph $G$, $T \subseteq V(G)$, $\counter\colon V(G) \rightarrow \mathds{N}$, $k \in \mathds{N}$, and~$t \in \mathds{Q}$.}
{Is there a vertex set~$S$ of size~$k$ such that~$T \subseteq S \subseteq V(G)$   and~$\val_G(S) \ge t$ (\textsc{Max}) or $\val_G(S) \le t$ (\textsc{Min}), respectively?}

A vertex set~$S$ fulfilling these requirements is referred to as a \emph{solution}.
Now, we define the \emph{contribution} of a vertex.
The contribution of a vertex~$v$ is a measure on how much the value of a partial solution~$T$ increases if~$v$ is added to~$T$.
Note that our definition slightly differs from that of Bonnet et al.~\cite{BEPT15}.

\begin{definition}
\label{def-contribution}
For a vertex set $T \subseteq V(G)$, we define the \emph{contribution} of a vertex~$v$ as
\begin{align*}
	\contribution(v, T) \coloneqq {}& \alpha \cdot (|N(v) \setminus T| + \counter(v)) + (1 - 2 \alpha) |N(v) \cap T| \\ 
	={}& \alpha \degCounter(v) +  (1 - 3 \alpha) |N(v) \cap T|.
\end{align*}
\end{definition}

The contribution is chosen so that the value~$\val(S)$ of a vertex set~$S$ computes as follows. 
\begin{lemma}
	\label{lemma:val}
	Let~$G$ be a graph and~$S \coloneqq \{ v_1, \dots, v_\ell \} \subseteq V(G)$ a vertex set.
	Then, it holds that $\val(S) =  \sum_{i \in [\ell]} \contribution(v_i, \{ v_1, \dots, v_{i - 1} \})$.
\end{lemma}

\begin{proof}
  Let $S_i = \{ v_1, \dots, v_{i - 1} \}$ and $\overline{S_i} = \{ v_i, \dots, v_{\ell} \}$ for each $i \in [\ell]$.
  Observe that
  \begin{align*}
    m(S, V(G) \setminus S) &= \sum_{i \in [\ell]} |N(v_i) \setminus S| = \sum_{i \in [\ell]} |N(v_i) \setminus S_{i - 1}| - |N(v_i) \cap \overline{S_i}| \text{ and } \\
    m(S) &= \sum_{i \in [\ell]} |N(v_i) \cap S_{i - 1}| = \sum_{i \in [\ell]} |N(v_i) \cap \overline{S_i}|.
  \end{align*}
  We thus have
  \begin{align*}
    \val(S) &= \alpha (m(S, V(G) \setminus S) + m(S) + \counter(S)) + (1 - 2\alpha)\cdot m(S) \\
    &= \sum_{i \in [\ell]} \bigg( \alpha [ (|N(v_i) \setminus S_{i - 1}| - |N(v_i) \cap \overline{S_i}|) + |N(v_i) \cap \overline{S_i}| + \counter(v)] \\ 
    &\qquad\qquad + (1 - 2\alpha) |N(v_i) \cap S_{i - 1}| \bigg)\\
    &=\sum_{i\in[\ell]}\contribution(v_i,\overline{S_i}) = \sum_{i\in[\ell]}\contribution(v_i,S_i),
  \end{align*}
  which proves the lemma.
\end{proof}

For a vertex $v$ and two sets $X \subseteq Y \subseteq V(G)$, we have $\contribution(v, X) \ge \contribution(v, Y)$ for $\alpha \in (1/3, 1]$ and $\contribution(v, X) \le \contribution(v, Y)$ for $\alpha \in [0, 1/3)$ by \Cref{def-contribution}.
Note that a function~$f:2^P\to \mathds{Q}$ such that for each~$X,Y\subseteq P$ with~$X\subseteq Y$ and for each element~$v\in P\setminus Y$ it holds that~$f(X\cup\{v\})-f(X)\ge f(Y\cup\{v\})-f(Y)$ is called \emph{submodular} and a function~$g:2^P\to \mathds{Q}$ such that for each~$X,Y\subseteq P$ with~$X\subseteq Y$ and for each element~$v\in P\setminus Y$ it holds that~$g(X\cup\{v\})-g(X)\ge g(Y\cup\{v\})-g(Y)$ is called \emph{supermodular}.
By \Cref{lemma:val} we have~$\val_G(X\cup\{v\})=\val_G(X)+\contribution(v,X)$.
Hence, we conclude the following.

\begin{observation}
       \label{obs-val-is-submodular}
       The function~$\val_G(\cdot)$ is submodular for~$\alpha\in(1/3,1]$ and supermodular for~$\alpha\in[0,1/3)$.
\end{observation}

 \subsection{Main Reduction Rules \& Basic Exchange Argument}
Annotations are helpful for data reductions in the following way:
If we identify a vertex~$v$ that is (or is not) in a solution, then, we can simplify the instance as follows using the annotations.

\begin{rrule}[Inclusion Rule]
	\label{rr:general:how-to-include}
	If there is a solution $S$ with $v \in S \setminus T$, then add~$v$ to~$T$.
	If there is a vertex~$v \in T$ with~$\counter(v) > 0$, then decrease~$t$ by~$\alpha \cdot \counter(v)$ and set~$\counter(v) \coloneqq 0$. 
\end{rrule}

\begin{rrule}[Exclusion Rule]
	\label{rr:general:how-to-exclude}
	If there is a solution $S$ with $v \notin S$, then for each~$u\in N(v)$ increase~$\counter(u)$ by one and remove~$v$ from~$G$.
\end{rrule}

The correctness of these two rules follows by the definitions of a partial solution and of the counter.
Notice that we maintain the aforementioned invariant that every vertex $v \in T$ has $\counter(v) = 0$ when applying the Inclusion Rule (\cref{rr:general:how-to-include}).
Furthermore, we observe the following.

\begin{observation}
	\label{obs-rr-preserve-params}
	The Inclusion Rule (\cref{rr:general:how-to-include}) and the Exclusion Rule (\cref{rr:general:how-to-exclude}) do not increase the parameters maximum degree~$\Delta$,~$c$-closure, degeneracy~$d$, vertex cover number~$\mathsf{vc},$ and~$h$-index.
\end{observation}

The reduction rules themselves are simple.
The difficulty lies of course in identifying vertices that are included in or excluded from some solution.
In the respective arguments, we use the following notion of better vertices.

 \subsection{Better Vertices} 
The following notion captures a situation that frequently appears in our arguments for the annotated problem variant and allows for simple exchange arguments (see \cref{lemma:replace}).

\begin{definition}
	A vertex~$v \in V(G)$ is \emph{better} than~$u \in V(G)$ with respect to a vertex set $T \subseteq V(G)$ if $\contribution(v, T) \ge \contribution(u, T)$ for the maximization variant (if $\contribution(v, T) \le \contribution(u, T)$ for the minimization variant).
	
	A vertex~$v \in V(G)$ is \emph{strictly better} than~$u \in V(G)$ if for all~$T \subseteq V(G)$ of size at most~$k$ we have $\contribution(v, T) \ge \contribution(u, T)$ for the maximization variant ($\contribution(v, T) \le \contribution(u, T)$ for the minimization variant).
\end{definition}

When we simply say that $v$ is better than $u$, we mean that $v$ is better than $u$ with respect to the empty set.
The following lemma immediately follows from \Cref{lemma:val}.
\begin{lemma}
  \label{lemma:replace}
  Let $S$ be a solution of an instance of \probAno{}.
  Suppose that there are two vertices $v \in S$ and $v' \notin S$ such that $v'$ is better than $v$ with respect to $S \setminus \{ v \}$ or~$v'$ is strictly better than~$v$.
  Then, $S' \coloneqq (S \setminus \{ v \}) \cup \{ v' \}$ is also a solution.
\end{lemma}
\begin{proof}
  We give a proof for the maximization variant; the minimization variant follows analogously.
  By \Cref{lemma:val}, we have $\val(S') = \val(S \setminus \{ v \}) + \contribution(v', S \setminus \{ v \}) \ge \val(S \setminus \{ v \}) + \contribution(v, S \setminus \{ v \}) = \val(S)$.
  Here, the inequality follows from the fact that~$v'$ is better than~$v$.
\end{proof}

Observe that the contribution of any vertex~$v$ differs from~$\alpha \degCounter(v)$ by at most~$|(1 - 3\alpha)k|$.
This observation allows us to identify some strictly better vertices in the following.
This is helpful when we wish to apply the second part of \Cref{lemma:replace} on strictly better vertices.


\begin{lemma}
	\label{lem:general:condition-strictly-better}
	Let~$u,v \in V(G)$. 
	Vertex~$v$ is strictly better than~$u$ if
	\begin{itemize}
		\item (Maximization:) $\alpha\degCounter(u) \le \alpha \degCounter(v) - |(1 - 3\alpha)k|$.
		\item (Minimization:) $\alpha\degCounter(u) \ge \alpha \degCounter(v) + |(1 - 3\alpha)k|$.
	\end{itemize}
\end{lemma}

{
\begin{proof}
	We give a proof for the maximization variant; the minimization variant follows analogously.
	By the definition of strictly better vertices, it suffices to show that~$\contribution(v, T) - \contribution(u, T) \ge 0$ for each $T \subseteq V(G)$ of size at most~$k$:
	\begin{align*}
		& \contribution(v, T) - \contribution(u, T) \\
		&= \alpha \degCounter(v) +  (1 - 3 \alpha) |N(v) \cap T| - \alpha \degCounter(u) -  (1 - 3 \alpha) |N(u) \cap T| \\
		&= \alpha (\degCounter(v) - \degCounter(u)) + (1 - 3 \alpha) (|N(v) \cap T| - |N(u) \cap T|) \\
		&\ge |(1 - 3\alpha)k| + (1 - 3 \alpha) (|N(v) \cap T| - |N(u) \cap T|) \ge |(1 - 3\alpha)k| + (1 - 3 \alpha)k  \ge 0.
	\end{align*}
	This completes the proof.
\end{proof}
}

 \subsection{Reduction to Non-annotated Variant}
The following two lemmas (for the maximization and the minimization variant, respectively) show that it is possible to remove annotations without blowing up the instance size.
However, the instance size after removing annotations will depend on $\Delta_{\overline{T}} \coloneqq \max_{v \in V(G) \setminus T} \deg(v)$ and $\Gamma \coloneqq \max_{v \in V(G) \setminus T} \counter(v) + 1$.
Note that the maximum degree $\Delta =\max_{v \in V(G)} \deg(V) \ge \Delta_{\overline{T}}$.
We obtain an upper bound on $\Gamma$ in terms of $k + \Delta$ in the next section.

\begin{lemma}
  \label{lemma:general:remove-ano}
  Given an instance $\mathcal{I}\coloneqq (G, T, \counter, k, t)$ of \probAnoMax{} with $\alpha \in (0, 1]$, we can compute an equivalent instance~$\mathcal{I}'$ of \probMax{} of size~$\Oh((\Delta_{\overline{T}} + \Gamma + \alpha^{-1}) \cdot |V(G)| + \alpha^{-1} k \cdot |T|))$ in polynomial time.
\end{lemma}

\begin{proof}
  We may assume that $G$ has at least $k$ vertices (otherwise no solution for \probAnoMax{} exists and thus the empty graph and the same value of~$k$ form a no-instance for \probMax{}).
  We construct an equivalent instance~$\mathcal{I}'\coloneqq (G', k, t')$ of \probMax{}.
  The graph $G'$ is obtained from $G$ as follows:
  \begin{enumerate}
  \item
      Add~$\counter(v) + \lfloor \alpha^{-1} \rfloor$ degree-one neighbors to every vertex $v \in V(G)$.
    \item
      Additionally, add~$\ell \coloneqq \Delta_{\overline{T}} + \Gamma + |\alpha^{-1} - 3| \cdot k + \lfloor \alpha^{-1} \rfloor$ degree-one neighbors to every vertex $v \in T$.
  \end{enumerate}
  We denote by $L_v$ the set of degree-one vertices added to vertex~$v\in V(G)$ and we denote by~$L\coloneqq \bigcup_{v\in V(G)}L_v$ the set of all newly added leaf vertices.
  To conclude the construction of~$\mathcal{I}'$, we set~$t' \coloneqq t + \alpha(\ell \cdot |T| + \lfloor \alpha^{-1} \rfloor \cdot k)$.
  Since~$G$ has at most~$\Delta_{\overline{T}} \cdot |V(G)|$~edges and we add $\Oh((\Gamma + \alpha^{-1})\cdot |V(G)| + (\Delta_{\overline{T}} + \Gamma + \alpha^{-1} k) \cdot |T|)$~edges, we see that~$G'$ has $\Oh((\Delta_{\overline{T}} + \Gamma + \alpha^{-1}) \cdot |V(G)| + \alpha^{-1} k \cdot |T|)$~edges.

  Next, we prove the equivalence between~$\mathcal{I}$ and~$\mathcal{I}'$.
  For a solution~$S$ of~$\mathcal{I}$, its value in~$G'$ is increased by~$\alpha \cdot \lfloor \alpha^{-1} \rfloor$ for every vertex in~$S$ and, additionally, by~$\alpha\cdot\ell$ for every vertex in $T$, amounting to~$t + \alpha(\ell \cdot |T| + \lfloor \alpha^{-1} \rfloor \cdot k)$.

  Conversely, consider a solution~$S'$ of~$\mathcal{I'}$.
  First, we show that there is a solution containing all vertices of~$T$ and no leaf vertex of~$L$ using \Cref{lemma:replace}.
  Suppose that for some vertex~$v \in V(G)$, one of its degree-one neighbors~$v' \in L_v$ is in~$S'$ but not~$v$ itself.
  We then have~$\contribution(v', S' \setminus \{ v \}) = \alpha$ and~$\contribution(v, S' \setminus \{ v \}) \ge \alpha$, implying that~$(S' \setminus \{ v' \}) \cup \{ v \}$ is also a solution by \Cref{lemma:replace}.
  Thus, in the following we can assume that~$S' \cap L_v = \emptyset$ for every vertex~$v \in V(G)\setminus S'$.
  If there is a vertex~$v' \in S' \cap L_v$ for some~$v \in S'$, then by the assumption that~$|V(G)| \ge k$, the pigeonhole principle gives us a vertex~$w \in V(G) \setminus S'$ with~$S' \cap L_w = \emptyset$.
  Since~$|L_w| \ge \counter(w) + \lfloor \alpha^{-1} \rfloor \ge \lfloor \alpha^{-1} \rfloor$, we have~$\contribution(w, S' \setminus \{ v' \}) \ge \alpha \cdot \lfloor \alpha^{-1} \rfloor \ge \alpha (\alpha^{-1} - 1) = 1 - \alpha$.
  We thus have~$\contribution(v', S \setminus \{ v' \}) = 1 - \alpha \le \contribution(w, S' \setminus \{ v' \})$.
  Hence,~$(S' \setminus \{ v' \}) \cup \{ w \}$ is a solution, again by \Cref{lemma:replace}.
  
  Thus, we may assume that~$S'$ consists only of vertices from~$V(G)$.
  Suppose that some vertex~$v \in T$ is not in~$S'$.
  For any vertex~$v' \in S' \setminus T$, we have~$\deg(v') \le \deg_G(v) + \counter(v) + \lfloor \alpha^{-1} \rfloor \le \Delta_{\overline{T}} + \Gamma + \lfloor \alpha^{-1} \rfloor$.
  So we have $\deg(v') \ge \Delta_{\overline{T}} + \Gamma + |\alpha^{-1}| \cdot k + \lfloor \alpha^{-1} \rfloor \ge \deg(v) + |\alpha^{-1} - 3| \cdot k$.
  Applying \Cref{lem:general:condition-strictly-better} with $\counter(v) = \counter(v') = 0$, we obtain that~$v$ is strictly better than~$v'$.
  Now, it follows from \Cref{lemma:replace} that~$\mathcal{I'}$ has a solution~$S'$ such that $T \subseteq S' \subseteq V(G')$.
  Hence,~$S'$ is also a solution for~$\mathcal{I}$.
\end{proof}

\begin{lemma}
  \label{lemma:general:remove-ano-min}
  Given an instance~$\mathcal{I}\coloneqq (G, T, \counter, k, t)$ of \probAnoMin{} for $\alpha \in (0,1]$, we can compute an equivalent instance~$\mathcal{I'}$ of \probMin{} of size~$\Oh(\alpha^{-2} (\Delta + \Gamma + k)^2 + \alpha^{-1} (\Delta + \Gamma + k) \cdot |V(G)|)$ in polynomial time.
\end{lemma}

\begin{proof}
  We may assume that~$G$ has at least~$k$ vertices (otherwise no solution for \probAnoMin{} exists and the thus the empty graph and parameter~$k$ are a no-instance for \probMin{}).
  We construct an equivalent instance $\mathcal{I}'\coloneqq (G', k, t')$ of \probMin{}.
  
  Let~$\ell$ be the smallest integer greater than~$\alpha^{-1} (\Delta + \Gamma + |(1 - 3 \alpha)k|)$.
  Let~$G'$ be the graph obtained from~$G$ as follows:
  We add a clique $C$ on $2 \ell + 1$ vertices.
  For every vertex~$v \in V(G) \setminus T$, choose~$\ell + \counter(v)$ vertices of $C$ arbitrarily and add edges between $v$ and the chosen vertices.
  To conclude the construction of~$\mathcal{I}'$, we set~$t' \coloneqq t + \alpha \ell (k - |T|)$.
  Observe that we add $\Oh(\alpha^{-1} (\Delta + \Gamma + k))$~vertices and $\Oh(\alpha^{-2} (\Delta + \Gamma + k)^2 + \alpha^{-1} (\Delta + \Gamma + k) \cdot |V(G)|)$~edges.

  Next, we show that~$\mathcal{I}$ and~$\mathcal{I}'$ are equivalent.
  For a solution~$S$ of~$\mathcal{I}$ its value in~$G'$ is increased by~$\alpha\ell$ for every vertex~$v\in V(G)\setminus T$, amounting to~$t+\alpha\ell(k-|T|)$.
  Thus,~$S$ is also a solution of~$I'$.
  
  Conversely, suppose that~$\mathcal{I}'$ has a solution~$S'$.
  We show that there is a solution that contains all vertices of~$T$ and no vertex of~$C$.
  By construction, the following holds:
  \begin{enumerate}
    \item
      $\deg_{G'}(v) = \deg_G(v) \le \Delta$ for any vertex~$v \in T$.
    \item
      $\deg_{G'}(v) = \deg_{G}(v) + \ell + \counter(v) \in [\ell, \Delta + \Gamma + \ell]$ for any vertex~$v \in V(G) \setminus T$.
    \item
      $\deg_{G'}(v) \ge 2 \ell$ for any vertex~$v \in C$.
  \end{enumerate} 
  Since $\ell \ge \alpha^{-1}(\Delta + \Gamma + |(1 - 3 \alpha) k|)$, any vertex in $T$ is strictly better than any vertex in $V(G) \setminus T$ and any vertex in $V(G) \setminus T$ is strictly better than any vertex in $C$ by \Cref{lem:general:condition-strictly-better}:
  To see the latter, consider~$v_2 \in V(G) \setminus T$ and~$v_3 \in C$.
  Then we have:
  \begin{align*}
	\alpha \deg_{G'}(v_3) - \alpha \deg_{G'}(v_2) & \ge \alpha 2 \ell - \alpha (\Delta + \Gamma + \ell) \\
	& = \alpha \ell - \alpha \Delta - \alpha \Gamma \ge \Delta + \Gamma + |(1 - 3\alpha)k| - \alpha \Delta - \alpha \Gamma \ge |(1 - 3\alpha)k|.
  \end{align*}
  Thus, by \Cref{lemma:replace}, $\mathcal{I}'$ admits a solution~$S'$ with~$T \subseteq S' \subseteq V(G')$ and, hence,~$S'$ is a solution of value at least~$t' - \alpha \ell(k - |T|) = t$ for~$\mathcal{I}$.
\end{proof}

\subsection{Dependence of the Kernel Sizes on~$\alpha$}
To simplify notation, we will generally omit the polynomial factors in $\alpha^{-1}$ in the following sections.
Note that when we remove the annotations using \Cref{lemma:general:remove-ano} or \Cref{lemma:general:remove-ano-min}, a factor polynomial in $\alpha^{-1}$ appears in the size of the graph of the resulting \prob{} instance.
In our kernelization, we apply \Cref{lemma:general:remove-ano} or \Cref{lemma:general:remove-ano-min} once after obtaining an instance of \probAno{} in which the maximum degree~$\Delta$, the maximum counter~$\Gamma$, and the graph~$G$ are all bounded by some (polynomial) function of the parameter in question.

In Sections~\ref{sec:maxdeg} (maximum degree),~\ref{sec:closure} (closure), and~\ref{sec:degeneracy} (degeneracy) the bound on~$\Delta$ and the size of~$G$ will not depend on~$\alpha^{-1}$, while~$\Gamma$ has a term linearly dependent on $\alpha^{-1}$ (see \Cref{lemma:general:counter-bound}).
Thus, the kernel size will be proportional to $\alpha^{-1}$ for the maximization variant (we remark that $\alpha^{-1} \le 3$ for the degrading case) and $\alpha^{-4}$ for the minimization variant in the worst case.
In \Cref{sec:vc} (vertex cover number and~$h$-index), we make the dependence on $\alpha^{-1}$ explicit since many results have their own approach to obtain an upper bound on $\Delta$ and $\Gamma$.

%
%

\section{Parameterization By Maximum Degree}\label{sec:maxdeg}

\subsection{Polynomial Kernels in Degrading Cases}

Now, we present our framework to provide polynomial kernels of size~$\Delta+k$.
For this, it is essential to bound the largest counter of any vertex polynomial in~$\Delta+k$.
We do this by first adding vertices with a contribution far above~$t/k$ for maximization (and for below~$t/k$ for minimization) to the partial solution.
Second, we remove vertices with contribution far below~$t/k$ for maximization (and far above~$t/k$ for minimization).
We then show that this is sufficient to bound the counters.

To obtain a polynomial kernel for \prob{} with respect to~$\Delta+k$, we then show that it is sufficient to remove a vertex~$v$ if polynomial in~$\Delta+k$ many vertices are better than~$v$.
To obtain kernels for the smaller parameters $c$-closure (\Cref{sec:closure}) and degeneracy (\Cref{sec:degeneracy}) plus~$k$ it then remains to show that the maximum degree can be bounded in the parameter plus~$k$.

Recall that in the degrading cases we have~$\alpha\in(1/3,1]$ for  maximization and~$\alpha\in [0,1/3)$ for minimization.
Furthermore, recall that for two vertices~$u$ and~$v$,~$v$ is said to be better than~$u$ with respect to~$T$ if~$\contribution(v, T) \ge \contribution(u, T)$ (vice versa for the minimization variant).

Recall that we defined $\Delta_{\overline{T}} \coloneqq \max_{v \in V(G) \setminus T} \deg(v)$ for the annotated version. 
To obtain the kernels in this section it is sufficient to use~$\Delta$ instead of~$\Delta_{\overline{T}}$ in \Cref{rr:delta:better}.
However, in \Cref{sec:vc} it is sometimes important to use~$\Delta_{\overline{T}}$ in \Cref{rr:delta:better} to obtain the kernels.

\begin{rrule}
	\label{rr:delta:better}
	Let~$\mathcal{I}$ be an instance of \probAnoBor{}.
	If there are at least~$(\Delta_{\overline{T}}+1)( k-1)+1$ vertices that are better than~$v$ with respect to~$T$, then apply the Exclusion Rule (\Cref{rr:general:how-to-exclude}) to~$v$.
\end{rrule}

{
\begin{lemma}
\Cref{rr:delta:better} is correct.
\end{lemma}
\begin{proof}
	Let~$\mathcal{I}'$ be the reduced instance. 
	Clearly, if~$\mathcal{I}'$ has a solution~$S'$, then~$S'$ is also a solution for~$\mathcal{I}$.
	Conversely, suppose that~$\mathcal{I}$ has a solution~$S$.
	If~$v\notin S$, then~$S$ is also a solution for~$\mathcal{I'}$.
	In the following, we assume that~$v\in S$. 

	By the pigeonhole principle, there exists a vertex $v'$ better than $v$ such that $v \notin N[S\setminus T]$.
	We claim that $S' \coloneqq (S \setminus \{ v \}) \cup \{ v' \}$ is a solution for~$\mathcal{I'}$.
	First, we consider maximization.
	By \Cref{lemma:replace}, it suffices to show that $\contribution(v', S \setminus \{ v \}) \ge \contribution(v, S \setminus \{ v \})$.
	Since $v' \notin N[S' \setminus T]$, we have $\contribution(v', S \setminus \{ v \}) = \contribution(v', T) \ge \contribution(v, T) \ge \contribution(v, S \setminus \{ v \})$.
	Here, the last inequality follows from the fact that the contribution is \borderFocusedS{}.
	
	Second, we consider minimization.
	By \Cref{lemma:replace}, it suffices to show that $\contribution(v', S \setminus \{ v \}) \le \contribution(v, S \setminus \{ v \})$.
	Since $v' \notin N[S' \setminus T]$, we have $\contribution(v', S \setminus \{ v \}) = \contribution(v', T) \le \contribution(v, T) \le \contribution(v, S \setminus \{ v \})$.
	Here, the last inequality follows from the fact that the contribution is \borderFocusedS{}.
\end{proof}
}

Next, we show that the exhaustive application of \Cref{rr:delta:better} yields a polynomial kernel for \probBor{}.

\begin{proposition}
  \label{prop:general:delta}
  \probBor{} has a kernel of size
  \begin{itemize}
	\item $\Oh(\Delta^2 k)$ for maximization and~$\alpha\in(1/3,1]$, and
	\item $\Oh(\Delta k(\Delta +k))$ for minimization and~$\alpha\in (0, 1/3)$.
  \end{itemize}
\end{proposition}

\begin{proof}
  Given an instance of \probBor{}, we transform it into an equivalent instance of \probAnoBor{} and apply \Cref{rr:delta:better} exhaustively.
  Observe that since \cref{rr:delta:better} is applied, we have~$|V(G)|  \le \Delta k + 1$.
  Moreover, we have~$T = \emptyset$ and~$\Gamma \le \Delta$ since each neighbor of a vertex can increase its counter by at most one.
  By \Cref{lemma:general:remove-ano} (maximization) or \Cref{lemma:general:remove-ano-min} (minimization), we obtain an equivalent instance of \prob{} of size $\Oh(\Delta^2  k)$ (maximization) or $\Oh(\Delta k(\Delta +k))$ (minimization).
\end{proof}

Note that \Cref{prop:general:delta} does not cover the case $\alpha = 0$ for minimization, which is also called \textsc{Sparsest~$k$-Subgraph}.
The kernel for this case will be shown in \Cref{sec:degeneracy}: \Cref{prop-minimization-kernel-d+k-alpha=0} provides a kernel of size~$\Oh(d^2k)$ for \textsc{Sparsest~$k$-Subgraph}; since~$d\le\Delta$, this implies also a kernel of size~$\Oh(\Delta^2k)$.

\Cref{prop:general:delta} shows that given an instance of \probBor{}, we can find in polynomial time an equivalent instance of \probBor{} of size $\Oh(\Delta + k)^{O(1)}$.
In the following, in \Cref{prop:general:kernel}, we will show that an equivalent instance of \probBor{} that has size $(\Delta + k)^{\Oh(1)}$ can be constructed even if an instance of \probAnoBor{} is given.
\Cref{prop:general:kernel} plays an important role in kernelizations in subsequent sections.
Essentially, the task of kernelization for~$k+c$ and~$k+d$ boils down to bound the maximum degree~$\Delta$ to apply \Cref{prop:general:kernel}.

As shown in the proof of \Cref{prop:general:delta}, the number of vertices becomes polynomial in $k + \Delta$ by exhaustively applying \Cref{rr:delta:better}.
Recall that in \Cref{sec:annotation}, we presented a polynomial-time procedure to remove annotations with an additional polynomial factor in $\Delta + \Gamma$ on the instance size, where $\Gamma$ denotes the maximum counter.
To prove \Cref{prop:general:kernel}, it remains to bound~$\Gamma$ for \probAnoBor{}.

\paragraph{Bounding the largest counter $\Gamma$.}
Throughout the section, let~$k' \coloneqq k - |T|$ and~$t' \coloneqq t - \val(T)$.
First, we identify some vertices which are contained in a solution, if one exists.

\begin{definition}
	\label{def-border-focused-including}
	Let $\mathcal{I}$ be a yes-instance of \probAnoBor{}.
	A vertex~$v \in V(G) \setminus T$ is called \emph{satisfactory} if
	\begin{itemize}
		\item (Maximization:) $\contribution(v, T) \ge t' / k' + (3 \alpha-1)(k - 1)$ and~$\alpha\in(1/3,1]$.
		\item (Minimization:) $\contribution(v, T) \le t' / k' + (3 \alpha-1)(k - 1)$ and~$\alpha\in(0,1/3)$.
	\end{itemize}
\end{definition}

\begin{rrule}
	\label{rr:general:include-high}
	Let~$\mathcal{I}$ be an instance of \probAnoBor{} with~$\alpha>0$ and let~$v\in V(G)\setminus T$ be a satisfactory vertex.
	Apply the Inclusion Rule (\Cref{rr:general:how-to-include}) on vertex~$v$.
\end{rrule}

{

\begin{lemma}
	\label{lemma:general:include-high}
	\Cref{rr:general:include-high} is correct.
\end{lemma}

\begin{proof}
Let~$\mathcal{I}'$ be the reduced instance. 
	Clearly, if~$\mathcal{I}'$ has a solution~$S'$, then~$S'$ is also a solution for~$\mathcal{I}$.
	Conversely, suppose that~$\mathcal{I}$ has a solution~$S$.
  The lemma clearly holds for~$v \in S$.
  So we will assume that~$v \notin S$.
  We start with an auxiliary claim:

  \begin{claim}
    There is an ordering $(v_1, \dots, v_{k'})$ of the vertices of $S \setminus T$ with
    \begin{itemize}
    \item $\ell_1 \ge \ell_2 \ge \ldots \ge \ell_{k'}$ for maximization and~$\alpha\in(1/3,1]$, and
    \item $\ell_1 \le \ell_2 \le \ldots \le \ell_{k'}$ for minimization and~$\alpha\in(0,1/3)$,
\end{itemize} 
where~$S_i\coloneqq T \cup \{ v_1, \dots, v_{i - 1} \}$ and $\ell_i \coloneqq \contribution{}(v_i, S_i)$ for every $i \in [k']$.
  \end{claim}
  \begin{claimproof}
    Consider an ordering of~$S \setminus T$, where the~$i$-th vertex~$v_i$ is chosen in such a way that~$\contribution(v_i, S_i)$ is maximized (minimized).
    Then, since the contribution is \borderFocusedS{}, we conclude that~$\contribution(v_i, S_i) \ge \contribution(v_{i + 1}, S_i) \ge \contribution(v_{i + 1}, S_{i + 1})$ for maximization, and~$\contribution(v_i, S_i) \le \contribution(v_{i + 1}, S_i) \le \contribution(v_{i + 1}, S_{i + 1})$ for minimization, respectively,  for every~$i \in [k' - 1]$.
  \end{claimproof}

Note that~$S_{k'}=S\setminus\{v_{k'}\}$.
  Now, consider the vertex set $S' \coloneqq S_{k'} \cup \{ v \}$.
  We show that~$\val(S') \ge t$ for maximization and that~$\val(S') \le t$ for minimization.
  By \Cref{lemma:val}, we have~$\val(S') = \val(S) - \ell_{k'} + \contribution(v, S_{k'})$.
  The definition of contribution yields that
  \begin{align*}
    \contribution(v, S_{k'}) &=\alpha\degCounter(v)+(1-3\alpha)|N(v)\cap S_{k'}|\\
    &= \contribution(v, T) + (1 - 3 \alpha) |N(v) \cap (S_{k'} \setminus T)|.
  \end{align*}
  Since~$|S_{k'} \setminus T| = k' - 1$ and~$\contribution(v, T) \ge t' / k' + (3 \alpha-1)(k - 1)$ for maximization and~$\contribution(v, T) \le t' / k' + (3 \alpha-1)(k - 1)$ for minimization, we have
  \begin{align*}
    \contribution(v, S_{k'})
    &\ge [t' / k' + (3 \alpha-1)(k - 1)] + (1 - 3 \alpha) (k' - 1) \ge t' / k'\\ &\hspace{1cm}\text{ for maximization and } \alpha\in(1/3,1], \text{ and} \\
    \contribution(v, S_{k'})
    &\le [t' / k' + ( 3 \alpha-1)(k - 1)] + (1 - 3 \alpha) (k' - 1) \le t' / k'\\  &\hspace{1cm}\text{ for minimization and } \alpha\in(0,1/3).
  \end{align*}
  
  For maximization, if~$\ell_{k'} \le t' / k'$, then we have~$\ell_{k'} \le \contribution(v, S_{k'})$ and thus $\val(S') \ge \val(S) \ge t$.
  Analogously, for minimization, if~$\ell_{k'} \ge t' / k'$, then we have~$\ell_{k'} \ge \contribution(v, S_{k'})$ and thus $\val(S') \le \val(S) \le t$.
  Thus, in the following, we assume that~$\ell_{k'} > t' / k'$ for maximization and that~$\ell_{k'} < t' / k'$ for minimization.
  We obtain that
  \begin{align*}
    \val(S_{k'})
    &= \val(T) + \sum_{i = 1}^{k' - 1} \ell_{i} \text{. Thus,} \\
    \val(S_{k'}) &\ge \val(T) + (k' -  1) \ell_{k'} 
    \ge \val(T) + \frac{(k' - 1)t'}{k'} \text{ for maximization, and} \\
    \val(S_{k'}) &\le \val(T) + (k' -  1) \ell_{k'}
    \le \val(T) + \frac{(k' - 1)t'}{k'} \text{ for minimization.}
  \end{align*}
  Hence, for maximization~$\val(S') = \val(S_{k'}) + \contribution(v, S_{k'}) \ge [\val(T) + (k' - 1)t' / k'] + t' / k' = t$ and for minimization~$\val(S') = \val(S_{k'}) + \contribution(v, S_{k'}) \le [\val(T) + (k' - 1)t' / k'] + t' / k' = t$.
\end{proof}
}

We henceforth assume that \Cref{rr:general:include-high} is exhaustively applied on every satisfactory vertex.
Next, we identify some vertices which are not contained in any solution.

\begin{definition}
	\label{def-border-focused-excluding}
	Let $\mathcal{I}$ be a yes-instance of \probAnoBor{}.
	A vertex~$v \in V(G) \setminus T$ is called \emph{needless} if
	\begin{itemize}
		\item (Maximization:) $\contribution(v, T) < t' / k' - (3 \alpha-1)(k - 1)^2$ for maximization and~$\alpha\in(1/3,1]$.
		\item (Minimization:) $\contribution(v, T) > t' / k' - (3 \alpha-1)(k - 1)^2$ for minimization and~$\alpha\in(0,1/3)$.
	\end{itemize}
\end{definition}

\begin{rrule}
	\label{rr:general:exclude-low}
	Let~$\mathcal{I}$ be an instance of \probAnoBor{} with~$\alpha>0$ and let~$v\in V(G)\setminus T$ be a needless vertex.
	Apply the Exclusion Rule (\Cref{rr:general:how-to-exclude}) on vertex~$v$.
\end{rrule}

{

\begin{lemma}
\label{lemma:general:exclude-low}
\Cref{rr:general:exclude-low} is correct.
\end{lemma}

\begin{proof}
Let~$\mathcal{I}'$ be the reduced instance. 
	Clearly, if~$\mathcal{I}'$ has a solution~$S'$, then~$S'$ is also a solution for~$\mathcal{I}$.
	Conversely, suppose that~$\mathcal{I}$ has a solution~$S$.
  The lemma clearly holds for~$v \notin S$.
  So we will assume that~$v \in S$.
  Since \Cref{rr:general:include-high} is exhaustively  applied to each satisfactory vertex~$v'$, we obtain that~$\contribution(v', T) < t' / k' + (3 \alpha-1)(k - 1)$ for maximization and~$\contribution(v', T) > t' / k' + (3 \alpha-1)(k - 1)$ for minimization, for each vertex~$v'\in V(G)\setminus T$.
  Together with \Cref{lemma:val} we thus obtain for maximization that
  \begin{align*}
    \val(S) 
    &\le \val(T) + \contribution(v, T) + (k' - 1)[t' / k' + (3 \alpha-1)(k - 1)] \\ 
    &< \val(T) + [t' / k' -  (3 \alpha-1)(k - 1)^2] + (k' - 1)[t' / k' + (3 \alpha-1)(k - 1)]  \\
    &= \val(T)+t'-(3\alpha-1)(k-1)(k-k')\le t \text{ since } \alpha\in(1/3,1].
    \end{align*}
    Similarly, for minimization, we obtain that
    \begin{align*}
    \val(S) 
    &\ge \val(T) + \contribution(v, T) + (k' - 1)[t' / k' + ( 3 \alpha-1)(k - 1)] \\ 
    &> \val(T) + [t' / k' -  (3 \alpha-1)(k - 1)^2] + (k' - 1)[t' / k' + (3 \alpha-1)(k - 1)]  \\
    &= \val(T)+t'+(1-3\alpha)(k-1)(k-k')\ge t \text{ since } \alpha\in(0,1/3).
  \end{align*}
For maximization, this is a contradiction to~$\val(S)\ge t$, and for minimization this is a contradiction to~$\val(S)\le t$.
Thus,~$S$ cannot contain the needless vertex~$v$.
\end{proof}
}

We henceforth assume that \Cref{rr:general:exclude-low} is applied on every needless vertex.
The following reduction rule decreases the counter of each vertex in~$V(G)\setminus T$.
After this rule is exhaustively applied, we may assume that~$\counter(v) = 0$ for at least one vertex~$v \in V(G) \setminus T$.
Recall that we already have~$\counter(v)=0$ for every vertex in~$T$.

\begin{rrule}
  \label{rr:general:no-zero-counter}
  If $\counter(v) > 0$ for every vertex~$v \in V(G) \setminus T$, then decrease~$\counter(v)$ by~1 for every vertex~$v \in V(G) \setminus T$ and decrease~$t$ by~$\alpha k'$.
\end{rrule}

Next, we show that after the exhaustive application of \Cref{rr:general:no-zero-counter} the counter of each vertex is bounded polynomially in terms of~$\Delta$ and~$k$.

\begin{lemma}
	\label{lemma:general:counter-bound}
	Let~$\mathcal{I}$ be a reduced yes-instance of \probAnoBor{} with~$\alpha >0$.
	We have $\counter(v) \in \Oh(\Delta + \alpha^{-1} k^2)$ for every vertex $v \in V(G) \setminus T$.
\end{lemma}

{
\begin{proof}
	First, observe that there exists at least one vertex~$u \in V(G) \setminus T$ with $\counter(u) = 0$, since otherwise \Cref{rr:general:no-zero-counter} is still applicable.
	Since \Cref{rr:general:exclude-low} is applied to each needless vertex, we conclude that every vertex in~$V(G)\setminus T$ has contribution at least~$t' / k' - (3\alpha-1)(k - 1)^2$ for maximization.
	Furthermore, since \Cref{rr:general:include-high} is applied to each satisfactory vertex, we conclude that every vertex~$v\in V(G)\setminus T$ has contribution at least~$t' / k' + ( 3 \alpha-1)(k - 1)$ for minimization.
	In particular, we have
	\begin{align*}
		\contribution(u, T)&\ge t' / k' - (3\alpha-1)(k - 1)^2 \text{ for maximization, and} \\
		\contribution(u, T)&\ge t' / k' + ( 3\alpha-1)(k - 1) \text{ for minimization}.
	\end{align*}
	Since also~$\contribution(u, T) = \alpha \cdot \deg(u) + (1 - 3 \alpha) |N(u) \cap T|$ we obtain that
	\begin{align}
		\label{eq:lb1}
		t' / k' &\le \alpha \cdot \deg(u) + (1 - 3 \alpha)|N(u) \cap T| + (3 \alpha-1)(k - 1)^2 \text{ for maximization, and} \\
		\label{eq:lb2}
		t' / k' &\ge \alpha \cdot \deg(u) + (1 - 3 \alpha)[(k - 1) + |N(u) \cap T|] \text{ for minimization}.
	\end{align}
	Moreover, since \Cref{rr:general:how-to-include} is applied to each satisfactory vertex, we conclude that every vertex~$v\in V(G)\setminus T$ has contribution at most~$t' / k' + (3 \alpha-1)(k - 1)$ for maximization.
	Furthermore, since \Cref{rr:general:how-to-exclude} is applied to each needless vertex, we conclude that every vertex in~$V(G)\setminus T$ has contribution at most~$t' / k'- (3\alpha-1)(k - 1)^2$ for minimization.
	This implies that in particular
	\begin{align}
		\label{eq:ub1}
		\alpha \cdot \counter(v) &\le t' / k' + (3 \alpha-1)(k - 1) \text{ for maximization, and}. \\
		\label{eq:ub2}
		\alpha \cdot \counter(v) &\le t' / k' - (3 \alpha-1)(k - 1)^2 \text{ for minimization}.
	\end{align}
	For maximization and~$\alpha\in(1/3,1]$ it then follows from \Cref{eq:lb1,eq:ub1} that
	\begin{align*}
		\counter(v) &\le \deg(u) + \frac{3 \alpha-1}{\alpha} [k(k - 1) - |N(v) \cap T|] \in \Oh(\Delta + \alpha^{-1}k^2).
		\end{align*}
		For minimization and~$\alpha\in(0,1/3)$ it then follows from \Cref{eq:lb2,eq:ub2} that
	\begin{align*}
		\counter(v) &\le \deg(u) + \frac{1-3 \alpha}{\alpha} [k(k - 1) + |N(v) \cap T|] \in \Oh(\Delta + \alpha^{-1}k^2).
	\end{align*}
	This concludes the proof.
\end{proof}
}

\paragraph{Putting everything together.}

For the kernels, we first transform the instance into an equivalent instance of \probAnoBor{}.
Second, we apply our reduction rules.
For the third step, we use the following proposition to reduce back to the non-annotated version.

\begin{proposition}
	\label{prop:general:remove-ano-n-bounded}
	Let~$\alpha > 0$.
	\begin{itemize}
	\item Given an instance~$(G, T, \counter, k, t)$ of \probAnoBorMax{}, we can compute in polynomial time an equivalent instance of \probBorMax{} of size $\Oh(|V(G)|^2 + \alpha^{-1}|V(G)|k^2) \subseteq \Oh(\alpha^{-1}|V(G)|^3)$.
	\item Given an instance~$(G, T, \counter, k, t)$ of \probAnoBorMin{}, we can compute in polynomial time an equivalent instance of \probBorMin{} of size $\Oh(\alpha^{-2}(|V(G)|+\alpha^{-1}k^2)^2) \subseteq \Oh(\alpha^{-4}|V(G)|^4)$.
	\end{itemize}
\end{proposition}

{
\begin{proof}
 	Using \Cref{rr:general:how-to-exclude,rr:general:how-to-include,rr:general:no-zero-counter} exhaustively yields, by \Cref{lemma:general:counter-bound}, an instance where $\counter(v) \in \Oh(\Delta+\alpha^{-1}k^2)\subseteq \Oh(|V(G)|+\alpha^{-1}k^2)$.
 	For maximization, by \Cref{lemma:general:remove-ano} we get an equivalent instance of \probBorMax{} of size $$\Oh(|V(G)|^2 + |V(G)|k^2+k^2) \subseteq \Oh (\alpha^{-1}|V(G)|^3), \text{ and}$$
 	for minimization, by \Cref{lemma:general:remove-ano-min} we get an equivalent instance  of \probBorMin{} of size 
\begin{align*}
\Oh(\alpha^{-2}(|V(G)|+\alpha^{-1}k^2)^2+\alpha^{-1}(|V(G)|+\alpha^{-1}k^2)|V(G)|) &\subseteq \Oh(\alpha^{-2}(|V(G)|+\alpha^{-1}k^2)^2)\\ &\subseteq \Oh(\alpha^{-4}|V(G)|^4).
\end{align*} 	
 	Hence, the statement follows.
\end{proof}
}

\begin{proposition}
	\label{prop:general:kernel}
	Given an instance~$(G, T, \counter, k, t)$ of \probAnoBor{} with~$\alpha> 0$, we can compute in polynomial time an equivalent \probBor{} instance of size~$(\Delta + k)^{\Oh(1)}$.
\end{proposition}

\begin{proof}
First, we bound~$\counter(v)$ by~$(\Delta+k)^{\Oh(1)}$ due to \Cref{lemma:general:counter-bound}.
Second, we obtain a kernel of size~$(\Delta+k)^{\Oh(1)}$ due to \Cref{prop:general:delta}.
Finally, we transform the resulting instance in an equivalent instance of \probBor{} of size~$(\Delta+k)^{\Oh(1)}$ due to \Cref{lemma:general:remove-ano} (maximization) or \Cref{lemma:general:remove-ano-min} (minimization).
\end{proof}

\subsection{No Polynomial Kernels in Non-Degrading Cases}

Note that if~$\alpha=0$, then \probMax{} corresponds to \textsc{Densest~$k$-Subgraph}.
It is already known that \textsc{Densest~$k$-Subgraph} does not admit a polynomial kernel for~$\Delta+k$~\cite{KS15}.
We strengthen and generalize this result: First, we observe that \textsc{Densest~$k$-Subgraph} does not admit a polynomial kernel for~$k$, even when~$\Delta = 3$. Second, we 
extend this negative result to \probIntMax{} and \probIntMin{} when~$\Delta$ is a constant.


\begin{theorem}
	\label{thm:maxdeg:no-poly-kernel-const-delta}
  Unless \PHC,
  \begin{enumerate}
  \item For each~$\alpha\in [0,1/3)$, \probIntMax{} on subcubic graphs does not admit a polynomial kernel for~$k$, and 
  \item For each~$\alpha\in (1/3,1]$, \probIntMin{}  on graphs with constant maximum degree does not
    admit a polynomial kernel for~$k$.
\end{enumerate}
\end{theorem}

{
\begin{proof}
  \newcommand{\nn}{\ensuremath{\hat{n}}}
  We present a polynomial-parameter transformation from  \textsc{Clique} parameterized by the size of the largest connected component (this parameterization does not admit a polynomial kernel which can be seen by an or-composition of disjoint \textsc{Clique} instances with the same clique size) to  \probIntMax{} (and \probIntMin{}) parameterized by the solution size~$k$.
  The reduction is based on a reduction of Feige and Seltser~\cite{Fei97} that shows NP-hardness of \textsc{Densest~$k$-Subgraph} on subcubic graphs with a subsequent reduction to    \probIntMax{} and \probIntMin{}, respectively.

  The reduction from \textsc{Clique} to \textsc{Densest~$k$-Subgraph} works as follows.
  Let~$(G,\ell)$ be an instance of \textsc{Clique}. Without loss of generality, we may
  assume that every connected component of~$G$ has exactly~$\nn$ vertices and more
  than~$\binom{\ell}{2}$ edges.  For each connected component~$H$ of~$G$ perform the
  following construction.  Let~$\{v_1, \ldots, v_{\nn}\}$ denote the vertex set
  of~$H$. For each vertex~$v_i$ of~$H$ create a cycle~$C_i\coloneqq (v_i^1,\ldots,v_i^{\nn})$ of
  length~\nn. This cycle is called the \emph{vertex cycle of}~$v_i$. 
  Then, add a path on~$\nn^2+1$ edges
  between~$v_r^q$ and~$v_q^r$ if~$v_r$ and~$v_q$ are adjacent in~$H$. 
In the following, this path is called \emph{connector path} of~$v^r_q$ and~$v^q_r$. 
  To each newly added vertex on
  this connector path, attach a degree-one vertex. 
  Furthermore, attach a degree-one vertex to each vertex on a vertex cycle which has no connector path.  
  We call these newly added degree-one vertices~\emph{pendant}. Observe
  that every non-pendant vertex has degree three.  Let~$G'$ denote the constructed graph.

  \begin{claim}\label{claim:clique-t-dks}
    The graph~$G$ contains a clique of size~$\ell$ if and only if~$G'$ contains a~$k'\coloneqq \ell\nn+\binom{\ell}{2}\nn^2$-vertex subgraph with at least~$\delta'\coloneqq \ell\nn + \binom{\ell}{2}(\nn^2+1)$ edges.
  \end{claim}
  \begin{claimproof}
    Suppose that~$G$ contains a clique~$K$ of size~$\ell$. Then, consider the induced subgraph that
    contains the~$\ell$ vertex cycles of the clique vertices and all connector paths between them. This subgraph has the claimed number of vertices
    and edges since the number of cycles is~$\ell$ and since each pair of cycles is
    connected via a path because~$K$ is a clique.

    Conversely, assume that~$G'$ contains a~$k'$-vertex subgraph~$G'[S]$ with at
    least~$\delta'$ edges.  First, observe that we may assume
    that~$S$ contains no pendant vertices: By construction and the fact that each connected component of~$G'$ has more than $\binom{\ell}{2}$~edges, each non-pendant vertex
    of~$G'$ is connected to more than~$k'$ non-pendant vertices. Hence,
    every pendant vertex of~$S$ can be replaced by a non-pendant neighbor~$u\notin S$ of some
    non-pendant vertex~$v\in S$ without decreasing the number of
    edges in~$G'[S]$.

    After this observation, the correctness proof can be carried out in the same manner as
    in the proof of Feige and~Seltser~\cite{Fei97}. We sketch the details for sake of
    completeness. First, observe that every degree-one vertex~$u$ of~$G'[S]$ may reach at
    most one degree-3 vertex~$v$ of~$G'[S]$ via some path that contains only vertices that
    have degree 2 in~$G'[S]$. We say that the degree-one vertex~$u$ is \emph{associated}
    with the degree-3 vertex~$v$. We call a degree-3 vertex \emph{good} if it is not
    associated with any degree-one vertex. A simple calculation shows that~$G'[S]$ must
    contain at least~$2\binom{\ell}{2}$ good degree-3 vertices to achieve the claimed
    number of edges~(refer to Feige and Seltser~\cite{Fei97} for further details). If a
    degree-3 vertex is good, then~$G'[S]$ contains its connector path completely. By the choice
    of~$k'$ and the path length,~$G'[S]$ may contain at most~$\binom{\ell}{2}$ connector paths
    completely and thus both endpoints of a connector path that is contained completely in~$G'[S]$
    are good. Now, let~$\mathcal{P}$ denote the set of connector paths that are completely contained
    in~$G'[S]$ and let~$\mathcal{C}$ denote the collection of vertex cycles that contain a
    good vertex or, equivalently, that contain an endpoint of a connector path in~$\mathcal{P}$. 
    
    We
    show that we may assume that every vertex cycle of~$\mathcal{C}$ is completely contained
    in~$G'[S]$. Assume, towards a contradiction, that~$G'[S]$ does not contain all vertex cycles
    of~$\mathcal{C}$ completely and that~$G'[S]$ is a~$k'$-vertex and~$\delta'$-edge
    subgraph of~$G$ with the largest possible number of vertices of~$\mathcal{C}$. Choose
    some vertex cycle~$C_v$ of~$\mathcal{C}$ that is not contained completely
    in~$G'[S]$. By definition of~$\mathcal{C}$, this vertex cycle contains a good
    vertex~$u$. Since~$C_v$ is not contained completely in~$G'[S]$, some vertex~$y$
    of~$C_v$ has degree one in~$G'[S]$ and is connected to~$u$ via a path of vertices
    of~$C_v$. Since~$u$ is good, there must be some vertex~$w$ on this path which has
    degree 3 in~$G'[S]$ but which is not good. Thus,~$G'[S]$ contains some but not all
    vertices of the connector path starting at~$w$. Let~$x$ be the degree-one vertex on
    this connector path that is contained in~$G'[S]$ and reachable from~$w$ via the connector path
    vertices. Removing~$x$ from~$G'[S]$ and adding the missing cycle neighbor of~$y$ gives
    a graph~$\tilde{G}$ with the same number of edges and vertices but with more cycle
    vertices, a contradiction to the choice of~$G'[S]$.

    Consequently, we may assume that~$G'[S]$ contains all cycles of~$\mathcal{C}$
    completely. By the choice of~$k'$ we have that~$|\mathcal{C}|=\ell$. Thus,~$G'[S]$
    contains~$\ell$ vertex cycles with~$\binom{\ell}{2}$~connector paths between them. The
    vertex set of~$G$ corresponding to the vertex cycles is thus a clique of size~$\ell$.
  \end{claimproof}
  
  The proof of \Cref{claim:clique-t-dks} also implies the correctness of the following observation. 
  \begin{claim}\label{claim:reg-sol}
    If~$G'$ contains an~$\ell\nn+\binom{\ell}{2}\nn^2$-vertex subgraph with at least~$\ell\nn + \binom{\ell}{2}\big(\nn^2+1\big)$ edges, then it contains one such subgraph that has no pendant vertices. 
    In other words, such a size~$\ell\nn+\binom{\ell}{2}\nn^2$-vertex subgraph contains only vertices of degree~$3$.
  \end{claim}
  
\paragraph{No polynomial kernel for \probIntMax{}}  
  Observe that the size parameter~$k'$ of the \textsc{Densest-$k$-Subgraph} instance
  depends only on~$\nn$ and that the maximum degree of~$G'$ is~$3$. Moreover,
  \Cref{claim:clique-t-dks} shows the correctness of the reduction. We thus have
  shown the theorem statement for \textsc{Densest-$k$-Subgraph}, the special case of
  \probIntMax{} with~$\alpha=0$.  We now provide a polynomial-time transformation from the
  \textsc{Densest-$k$-Subgraph} instance~$(G',k',\delta')$ to an instance~$(G',k',t)$ of
  \probIntMax{} with arbitrary~$\alpha\in [0,1/3)$ by
  setting~$t\coloneqq  (1-\alpha)\delta'+\alpha (3k'-2\delta')$.  Since~$G'$ and~$k'$ are not
  changed by this transformation, it only remains to show its correctness.

  Assume that~$(G',k',\delta')$ is a yes-instance, and let~$S$ be a size-$k'$ vertex set such that~$G'[S]$ has at least~$\delta'$ edges. By \Cref{claim:reg-sol}, we may assume that~$S$ contains only vertices that have degree 3 in~$G'$. Thus, in~$G'$, we have~$\val(S) = (1-\alpha)m(S)+\alpha m(S,V(G')\setminus S)=(1-\alpha)m(S)+\alpha (3k'-2m(S))\ge (1-\alpha)\delta'+\alpha (3k'-2\delta')$ since~$m(S)\ge \delta'$ and~$(1-\alpha)>2\alpha$.
  
  Conversely, let~$S$ be a size-$k'$ set with~$\val(S)\ge t=(1-\alpha)\delta'+\alpha (3k'-2\delta')$ in~$G'$. 
  First, observe that we may assume that~$S$ does not contain pendant vertices: 
  Each such vertex has degree one in~$G'$ and thus degree at most one in~$G'[S]$. 
  Such a vertex exists since each connected component of~$G$ has at least~$\binom{\ell}{2}$ edges and hence each connected component of~$G'$ has at least~$k'$ vertices.
  Consequently, it can be replaced by some non-pendant neighbor of a vertex in~$S$. Consequently, $\val(S) =(1-\alpha)m(S)+\alpha (3k'-2m(S))$. As in the proof of the forward direction, it follows that~$m(S)\ge \delta'$.

\paragraph{No polynomial kernel for \probIntMin{}} 
  Finally, we provide a polynomial-parameter transformation from the
  \textsc{Densest-$k$-Subgraph} instance~$(G',k',\delta')$ to an equivalent instance~$(G'',k',t)$ of
  \probIntMin{}. 
  In the construction, we distinguish whether $\alpha\in(1/3,1)$ or~$\alpha=1$. 
  In both cases, we add a gadget to each pendant vertex of~$G'$ such that each size-$k'$ set~$S$ with minimal value~$\val(S)$ in~$G''$
  contains only vertices that have degree at most 3 in~$G''$, which will be exactly the vertices of~$G'$.
  Hence, for both cases it remains to show that any solution with minimal value~$\val(S)$ contains only non-pendant vertices of~$G'$.
  Afterwards, the proof of the equivalence of~$(G,\ell)$ and~$(G'',k',t)$ is the same as for \probIntMax{}.

  First, we show the statement for~$\alpha\in(1/3,1)$.
  Construct~$G''$ from~$G'$ by adding for each pendant vertex~$u$
  in~$G'$ a set~$K_u$ on~$\lceil 15/(1-\alpha)\rceil$ vertices and making~$K_u\cup \{u\}$ a
  clique. Set~$t\coloneqq (1-\alpha)\delta'+\alpha (3k'-2\delta')$. 
    Assume towards a
  contradiction that some set~$S$ of minimal value contains a vertex~$u$
  with~$\deg(u) > 3$ in~$G''$. 
  By construction~$u$ is part of some clique~$K_y$ where~$y$ is some pendant vertex of~$G'$.
  Let~$Y\coloneqq K_y\cap S$ and let~$r\coloneqq |Y|$.
  We show that for~$S'\coloneqq (S\setminus Y)\cup Z$ for any vertex set~$Z\subseteq V(G')\setminus S$ consisting of $r$~non-pendant vertices we have~$\val(S')<\val(S)$, contradicting the minimality of~$\val(S)$.
  Observe that such a vertex set~$Z$ exists since each connected component of~$G'$ has at least~$\binom{\ell}{2}$ edges and thus each connected component of~$G''$ has at least~$k''$ non-pendant vertices.
  Since~$|K_y|=\lceil 15/(1-\alpha)\rceil$, we obtain that~$\val(S\setminus Y)\le \val(S)-(1-\alpha)\binom{r}{2}-\alpha r\lceil 15/(1-\alpha)\rceil$.
  Furthermore, observe that adding the vertices in~$Z$ increases the objective value by at most~$3r$ since each non-pendant vertex has degree~$3$ in~$G''$.
  Also, note that removing~$Y$ from~$S$ may result in a new outgoing edge of~$S$ which is incident with the neighbor of~$y$ in~$V(G')$.
  Hence, we obtain that
  
  $$\val((S\setminus Y)\cup Z)\le \val(S)-(1-\alpha)\binom{r}{2}-\alpha r\lceil 15/(1-\alpha)\rceil +3r +\alpha.$$
  
  Now, if~$r\le \lceil 15/(1-\alpha)\rceil-4$, then~$-\alpha r\lceil 15/(1-\alpha)\rceil +3r +\alpha<0$ and thus~$\val((S\setminus Y)\cup Z)< \val(S)$, a contradiction to the minimality of~$S$.
  Otherwise, if~$r\ge \lceil 15/(1-\alpha)\rceil-3$, then~$(1-\alpha)\binom{r}{2}>(1-\alpha)\lceil 11/(1-\alpha)\rceil \lceil 10/(1-\alpha)\rceil/2>55/(1-\alpha)>3\cdot\lceil 15/(1-\alpha)\rceil+\alpha$ and thus~$\val((S\setminus Y)\cup Z)< \val(S)$, a contradiction to the minimality of~$S$.
  Hence,~$S$ contains only vertices of~$G'$.

  Second, we show the statement for~$\alpha=1$ which corresponds to \textsc{Max $(k,n-k)$-Cut}.
	Construct~$G''$ from~$G'$ by adding for each degree-one vertex~$u$
  in~$G'$ a graph~$H$ to~$G''$ and making an arbitrary vertex~$h\in V(H)$ adjacent to~$u$.
  The graph~$H$ is a~$(|V(G')|,q, \rho)$-edge expander, that is,~$H$ has~$|V(G')|$ vertices, is~$q$-regular for some constant~$q$ and every vertex set~$T\subseteq V(H)$ of size at most~$|V(G')|/2$ fulfills~$m(T,V(H)\setminus T)\ge \rho q|T|$.
  We choose the~$(|V(G')|,q, \rho)$-edge expander~$H$ in such a way that~$\rho q\ge 4$.
  Note that such a graph~$H$ exists and can be constructed in polynomial-time~\cite{AB09,HLW06}.
  Finally, we set~$t\coloneqq (1-\alpha)\delta'+\alpha (3k'-2\delta')$.
  Assume towards a
  contradiction that some set~$S$ of minimal value contains a vertex~$u$
  with~$\deg(u) > 3$ in~$G''$. 
  By construction,~$u$ is part of some copy of~$H$ which was attached to a pendant (degree-one) vertex~$y$ of~$G'$.
  Let~$Y\coloneqq H\cap S$ and let~$r\coloneqq |Y|$.
  
  Since~$H$ is a~$(|V(G')|,q, \rho)$-edge expander and since~$k'<|V(G'')|/2$ we obtain that~$m(Y,V(G'')\setminus S)\ge 4r$.
  Since each vertex in~$G'$ has degree at most~$3$, we thus obtain for any set~$Z\subseteq V(G')$ of size~$r$ that~$\val((S\setminus Y)\cup Z)\le \val(S\setminus Y)-4r+3r<\val(S)$, a contradiction to the minimality of~$S$.
  Hence,~$S$ contains only vertices of~$G'$.
\end{proof}
}

\section{Parameterization by \texorpdfstring{$c$}{c}-Closure}
\label{sec:closure}

Now, we show that the maximum degree~$\Delta$ can be bounded by~$k^{\Oh(c)}$ in the degrading variant.
This then gives us a kernel of size~$k^{\Oh(c)}$.
To prove this result we rely on a polynomial bound on the Ramsey bound in $c$-closed graphs.
Then, we show that these kernels cannot be improved under standard assumptions.
Finally, we provide W[1]-hardness for the non-degrading variant.

 \subsection{A Tight \texorpdfstring{$k^{\Oh(c)}$}{k O(c)}-size Kernel for the Degrading Case}
We develop a tight kernel of size $k^{\Oh(c)}$ for the \borderFocusedS{} case.

\paragraph{A~$k^{\Oh(c)}$-size kernel for the \borderFocusedS{} case.}

To this end, we apply a series of reduction rules to obtain an upper bound of~$k^{\Oh(c)}$ on the maximum degree.
Then, the kernel of size $k^{\Oh(c)}$ follows from \Cref{prop:general:kernel}.
In order to upper-bound the maximum degree, we rely on a polynomial Ramsey bound for~$c$-closed graphs~\cite{KKS20}.

\begin{lemma}[{\cite[Lemma~$3.1$]{KKS20}}]
  \label{lemma:closure:ramsey}
  Any~$c$-closed graph~$G$ on at least~$R_c(q, b) \coloneqq (c - 1) \cdot \binom{b - 1}{2} + (q - 1)(b - 1) + 1$ vertices contains a clique of size~$q$ or an independent set of size~$b$.
  Moreover, a clique of size $q$ or an independent set of size $b$ can be found in polynomial time.
\end{lemma}


Using a similar approach as \cref{rr:delta:better} (but exploiting the~$c$-closure instead of the maximum degree) yields the following.

\begin{rrule}
	\label{rr:closure:better}
	Let~$\mathcal{I}$ be an instance of \probAnoBor{}.
	Let~$v \in V(G)$ be some vertex and let $X_v \subseteq N(v)$ be the set of vertices better than $v$.
	If $|X_v| > (c - 1) k$, then apply the Exclusion Rule (\Cref{rr:general:how-to-exclude}) to~$v$.
\end{rrule}

\begin{lemma}
	\label{lemma:closure:lowdeg}
	\cref{rr:closure:better} is correct.
\end{lemma}
\begin{proof}
	We provide a proof for the maximization version; the minimization version follows analogously.
	Let $S$ be a solution.
	Assume that $v \in S$ (we are done otherwise).
	We show that there is a vertex $v' \ne v$ such that $S' \coloneqq (S \setminus \{ v \}) \cup \{ v ' \}$ constitutes a solution.
	By \Cref{lemma:replace}, it suffices to show that $\contribution(v', S \setminus \{ v \}) \ge \contribution(v, S \setminus \{ v \})$.
	Let $S_v' \coloneqq S \setminus N[v]$.
	Each vertex in~$S_v'$ is, by definition, non-adjacent to $v$, and hence it shares at most $c - 1$ neighbors with $v$.
	This implies $|X_v \setminus N(S_v')| \ge |X_v| - (c - 1) \cdot |S_v'| > 0$ as~$X_v \subseteq N(v)$.
	Thus, there exists a vertex~$v' \in X_v \setminus N(S_v')$.
	Note in particular that $N(v') \cap S_v' = \emptyset$.
	Then, we have $N(v') \cap (S \setminus \{ v \}) \subseteq S \cap N(v)$ and thus $|N(v') \cap (S \setminus \{ v \})| \le |N(v) \cap (S \setminus \{ v \})|$.
	Moreover, we have $\alpha \degCounter(v') \ge \alpha \degCounter(v)$ (recall that $v'$ is better than $v$).
	Since~$\alpha \in (1/3, 1]$, it follows that
	\begin{align*}
		\contribution(v', S \setminus \{ v \}) &= \alpha \degCounter(v') + (1 - 3 \alpha) |N(v') \cap (S \setminus \{ v \})| \\
		&\ge \alpha \degCounter(v) + (1 - 3 \alpha) |N(v) \cap (S \setminus \{ v \})| = \contribution(v, S \setminus \{ v \}).
	\end{align*}
\end{proof}

Note that if there is a clique of size $(c - 1) k + 1$, then \cref{rr:closure:better} applies to one of the vertices with the smallest contribution.
Thus, applying \cref{rr:closure:better} exhaustively removes all cliques of size~$(c - 1) k + 1$.
By \Cref{lemma:closure:ramsey}, if there is a vertex $v$ with sufficiently large neighborhood, then we find a large independent set in $N(v)$.
We can then identify a vertex for which we can apply \Cref{rr:general:exclude-low}.

\begin{lemma}
  \label{lemma:pvc:findy}
  Suppose that $\Delta \ge R_c((c - 1)k + 1, (k + 1) k^{c - 2})$.
  Then, we can find in polynomial time a set $X$ of $i \in [c - 1]$ vertices and an independent set $I$ with the following properties:
  \begin{enumerate}
    \item
      The set $I \subseteq \bigcap_{x \in X} N(x)$ is an independent set of size at least $(k + 1) k^{c - i}$.
    \item
      For every vertex $u \in V(G) \setminus X$, it holds that $|N(u) \cap I| \le (k + 1) k^{c - i - 1}$.
  \end{enumerate}
\end{lemma}
\begin{proof}
  Let $v$ be a vertex such that $\deg(v) \ge R_c((c - 1)k + 1, (k + 1) k^{c - 2})$.
  Since there is no clique of size~$(c - 1)k + 1$, there is,  by \Cref{lemma:closure:ramsey}, an independent set $I_v$ of size~$(k + 1) k^{c - 2}$ in $N(v)$, which can be found in polynomial time.
  Let $X$ be an inclusion-wise maximal set of $i$ vertices including $v$ such that $|N^{\cap}(X) \cap I_v| > (k + 1) k^{c - i}$.
  Such a set can be found by the following polynomial-time algorithm:
  We start with~$X \coloneqq \{ v \}$ and $i \coloneqq 1$.
  We will maintain the invariant that $|X| = i$.
  If there exists a vertex $v' \in V(G) \setminus X$ with $|N(v') \cap N^\cap(X) \cap I_v| > (k + 1) k^{c - i - 1}$, then we add $v'$ to $X$ and increase $i$ by 1.
  We keep doing so until there remains no such vertex $v'$.

  We show that this algorithm terminates for $i = |X| \le c - 1$.
  Assume to the contrary that the algorithm continues for $i = c - 1$.
  We then have that $|N(v') \cap N^\cap (X) \cap I_v| > (k + 1) k^{c - i - 1} = k + 1 \ge 2$ for some vertex $v' \in V(G) \setminus X$.
  Since $I_v$ is an independent set, the set $N(v') \cap N^\cap (X) \cap I_v$ contains two non-adjacent vertices.
  Note, however, that these two vertices have at least $|X \cup \{ v \}| = c$ common neighbors, contradicting the $c$-closure of $G$.

  Finally, we show that the set $X$ found by this algorithm and $I \coloneqq N^\cap (X) \cap I_v$ satisfy the three properties of the lemma.
  We have $|N^\cap (x) \cap I_v| = |N(v') \cap N^\cap (X \setminus \{ v \})  \cap I_v| > (k + 1) k^{c - (i - 1) - 1} = (k + 1) k^{c - i}$, where $v'$ is the last vertex added to $X$.
  Moreover, since $X$ is inclusion-wise maximal, we have $|N(u) \cap I| = |N(u) \cap  N^\cap (X) \cap I_v| \le (k + 1)k^{c - i - 1}$ for every vertex $u \in V(G) \setminus X$.
\end{proof}

\begin{rrule}
	\label{rr:closure:independent-set}
	Let~$\mathcal{I}$ be an instance of \probAnoBor{}.
	Let~$X, I$ be as specified in \Cref{lemma:pvc:findy} and let $v \in I$ be a vertex such that every other vertex in $I$ is better than $v$.
	If $k \ge 2$, then apply the Exclusion Rule (\Cref{rr:general:how-to-exclude}) to~$v$.
\end{rrule}

\begin{lemma}
	\cref{rr:closure:independent-set} is correct.
\end{lemma}
\begin{proof}
	Again, we show the proof for the maximization variant; the minimization variant follows analogously.
	For the sake of contradiction, assume that every solution $S$ contains $v$.
	By \Cref{lemma:pvc:findy}, every vertex $u \in V(G) \setminus X$ has at most $(k + 1) k^{c - i - 1}$ neighbors in $I$.
	Moreover, since $I$ is an independent set, we have $|I \cap N[v']| = 1$ for every vertex $v' \in I$ (including $v$).
	For $S' \coloneqq S \setminus X$, we have
	\begin{align*}
		|I \setminus N[S']|
		&\ge |I| - \sum_{u \in S' \setminus \{ v \}} |I \cap N[u]| - |I \cap N[v]| \\
		&\ge (k + 1) k^{c - i} - (k - 1) (k + 1) k^{c - i - 1} - 1
		= k^{c - i} + k^{c - i - 1} - 1 > 0.
	\end{align*}
	Let $v'$ be an arbitrary vertex in $I \setminus N[S']$.
	We show that $\contribution(v', S \setminus \{ v \}) \ge \contribution(v, S \setminus \{ v \})$.
	By \Cref{lemma:replace}, this will imply that $(S \setminus \{ v \}) \cup \{ v' \}$ is a solution not containing $v$.
	Since $v$ and $v'$ are both adjacent to all vertices of $X$ and~$\alpha \in (1/3, 1]$, we have $|N(v) \cap (S \setminus \{ v \})| > |X \cap (S \setminus \{ v \})|$.
	We thus have
	\begin{align*}
		\contribution(v', S \setminus \{ v \})
		&= \alpha \degCounter(v') + (1 - 3 \alpha) |X \cap (S \setminus \{ v \})| \\
		&\ge \alpha \degCounter(v) + (1 - 3 \alpha) |X \cap (S \setminus \{ v \})| \\ 
		&\ge \alpha \degCounter(v) + (1 - 3 \alpha) |N(v) \cap (S \setminus \{ v \})| = \contribution(v, S \setminus \{ v \}).
	\end{align*}
	Here, the first inequality follows from the fact that $v'$ is better than $v$.
\end{proof}

By applying these reduction rules exhaustively, we obtain an instance with $\Delta \le R_c((c - 1)k + 1, (k + 1) k^{c - 2}) \in k^{\Oh(c)}$.
\Cref{prop:general:kernel} then leads to the following:

\begin{proposition}
  \label{thm:closure:kernel}
  \probBor{} has a kernel of size $k^{\Oh(c)}$.
\end{proposition}

\paragraph{Matching lower bounds.}
Next, we show  that the kernel provided in \Cref{thm:closure:kernel} cannot be improved under standard assumptions.

\begin{proposition}
\label{thm-lb-max-c-1/3-1}
  \probBorMax{} has no kernel of size~$\Oh(k^{c-3-\epsilon})$ unless \PHC.
\end{proposition}

{
We first prove the following technical lemma.
This lemma is useful to obtain an upper bound (maximization) or a lower bound (minimization) for~$\val(S)$ if the set of vertices (denoted as~$C^*$ in the lemma) in the \emph{instance choice gadget} is fixed and we only have to choose vertices~$D^*$ in the \emph{instance gadgets} (the union of all the instance vertices is~$D$ in the lemma).

\begin{lemma}
	\label{lem-alpha-inner-or-out}
	Let~$(G,k,t)$ be an instance of \probBor{} with a solution~$S$ fulfilling~$C^*\subseteq S\subseteq C^*\cup D$ where~$D$ is a set of vertices of the same degree.
	Let~$D^*\coloneqq S\cap D$. 
	Then:
	\begin{enumerate}
		\item For maximization,~$\val(S)$ is maximal if~$m(D^*)+m(D^*,C^*)$ is minimal.
		\item For minimization,~$\val(S)$ is minimal if~$m(D^*)+m(D^*,C^*)$ is minimal.
	\end{enumerate}
\end{lemma}

{
\begin{proof}
	Let~$z\coloneqq \deg(p)$ for each~$d\in D$, let~$D^*\coloneqq \{d_1,\ldots,d_\ell\}$, and let~$D_i\coloneqq \{d_j\mid j<i\}$.
	By \Cref{lemma:val} we obtain that
	\begin{align*}
	\val(S)&=\val(C^*)+\sum_{i=1}^\ell\contribution(d_i,C^*\cup D_i) \\
	&=\val(C^*) +\alpha z\ell+(1-3\alpha)\sum_{i=1}^\ell |N(d_i)\cap C^*|+|N(d_i)\cap D_i| \\
	&=\val(C^*)+ \alpha z\ell +(1-3\alpha)\big(m(D^*,C^*)+m(D^*)\big)
	\end{align*}
	Note that~$\val(C^*)+\alpha z\ell$ is a constant.
	Thus, for maximization in the degrading case we obtain that~$\val(S)$ is maximized if~$m(D^*,C^*)+m(D^*)$ is minimized since~$(1-3\alpha)<0$.
	Furthermore, for minimization in the degrading case we obtain that~$\val(S)$ is minimized if~$m(D^*,C^*)+m(D^*)$ is minimized since~$(1-3\alpha)>0$.
\end{proof}
}

\begin{proof}[Proof of \Cref{thm-lb-max-c-1/3-1}.]
	We provide a weak~$q$-composition from \textsc{Independent Set} on 2-closed graphs to \probBorMax{} in~$(q+3)$-closed graphs.
	Here, we assume that~$k>3q+\lceil\frac{\alpha}{3\alpha-1}\rceil$.
	
	\problemdef{Independent Set}
	{A graph~$G=(V,E)$ and an integer~$k$.}
	{Is there an independent set of size exactly~$k$?}

	Let~$[t]^q$ be the set of $q$-dimensional vectors whose entries are in~$[t]$.
	For a vector~$x\in[t]^q$ we denote by~$x_i$ the~$i$-th entry of~$x$.
	Assume that~$q\ge 2$ is a constant and that we are given exactly~$[t]^q$ instances~$\mathcal{I}_{x} \coloneqq (G_{x}, k)$ of \textsc{Independent Set} on~$2$-closed graphs. 
	Let $V_x \coloneqq V(G_x)$ for each $x \in [t]^q$ and let $D \coloneqq \bigcup_{x \in [t]^q} V_x$.
	We construct an equivalent instance~$(H,k',t')$ of \probBorMax{} as follows.

	\paragraph{Construction:} 
	First, for each~$x\in[t]^q$ we add the graph~$G_x$ to~$H$.
	In other words, we added the \emph{instance gadgets} to~$H$.
	We then add a clique~$C$ (the \emph{instance choice gadget}) consisting of~$tq$ vertices to~$H$.
	The vertices of~$C$ are denoted by~$w^i_j$ with~$i\in[q]$ and~$j\in[t]$.
	Furthermore, for each~$x\in[t]^q$ and $v\in V_x$, we add the edge~$vw^i_{x_i}$ for each~$i\in[q]$.
	Fix an integer~$\ell\ge t^q\cdot n+ (\alpha k+k+1)\cdot\alpha^{-1}$.
	We add leaf vertices so that $\deg(w^i_j)=\ell+\lceil(3\alpha-1)tq\cdot\alpha^{-1}\rceil$ for each vertex~$w^i_j$ and~$\deg(v)=\ell$ for each vertex $v \in D$.
	We denote the union of all these added leaf vertices by $L$.
	Finally, we set~$k'\coloneqq k+tq-q$ and 
	\begin{align*}
	t'&\coloneqq(1-\alpha)\frac{(tq-q)(tq-q-1)}{2}\\ 
	&\hspace{1cm} +\alpha\left[k\ell+(tq-q)\left(\ell+\Bigl\lceil\frac{(3\alpha-1)tq}{\alpha}\Bigr\rceil\right)-(tq-q)(tq-q-1)\right].
	\end{align*}

	\paragraph{Closure:} We show that~$H$ is~$(q + 3)$-closed.
	Since leaf vertices have degree one and~$C$ is a clique, we only have to consider non-adjacent vertex pairs where one vertex~$u$ is in~$D$ and the other vertex~$v$ is in~$C\cup D$. 
	Without loss of generality we assume that~$u\in V_x$ for some~$x\in[t]^q$.
	Recall that~$N(u)\subseteq V_x \cup C\cup L$ and~$u$ has exactly~$q$ neighbors in~$C$.
	First, consider the case~$v\in C$. 
	By construction we obtain from~$uv\notin E$ that~$u'v\notin E$ for each~$u'\in V_x$.
	Thus,~$|N(u)\cap N(v)|=q$. 
	Second consider the case~$v\in D$.
	If~$v\notin V_x$, then~$|N(u)\cap N(v)|\le q-1$.
	Otherwise, if~$v\in V_x$, we obtain~$|N(u)\cap N(v)|\le q+2$ from the fact that~$\mathcal{I}_x$ is 2-closed.

	\paragraph{Correctness:} In the following, we prove that there exists an independent set of size exactly~$k$ for some instance~$\mathcal{I}_x$ with~$x\in [t]^q$ if and only if there exists a vertex set~$S$ of size exactly~$k'$ in~$H$ such that~$\val(S) \ge t'$.

	Suppose that instance~$\mathcal{I}_x$ has an independent set~$I$ of size exactly~$k$ for some~$x\in[t]^q$.
	By~$C^*\coloneqq C\setminus \{w^i_{x_i} \mid i \in [q] \}$ we denote the non-neighbors of~$V_x$ in~$C$.
	Note that~$|C^*|=tq-q$.
	We show that~$S\coloneqq I\cup C^*$ is a solution of~$(H,k',t')$.
	Clearly,~$|S|=k+tq-q=k'$.
	Since no vertex of~$I$ is connected with any vertex in~$C^*$,~$I$ is an independent set, and since~$C^*$ is a clique of size~$tq-q$, we conclude that~$m_H(S)=(tq-q)(tq-q-1)/2$.
	Furthermore, since each vertex in~$I$ has degree~$\ell$, and since each vertex in~$C^*$ has degree~$\ell+\lceil(3\alpha-1)tq\cdot\alpha^{-1}\rceil$, we conclude that~$m_H(S, V(H) \setminus S)=k\ell+(tq-q)(\ell+\lceil(3\alpha-1)tq\cdot\alpha^{-1}\rceil)-(tq-q)(tq-q-1)$.
	Thus,~$\val(S)=t'$ and hence~$(H,k',t')$ is a yes-instance of \probBorMax{}.

	Conversely, suppose that~$(H,k',t')$ has a solution~$S\subseteq V(H)$ of size exactly~$k'$ with~$\val(S)\ge t'$. 
	First, we show that we can assume that~$S\cap L=\emptyset$.
	Assume that there exists a vertex~$v\in S\cap L$ and let~$w\in (C\cup D)\setminus S$.
	We show that for~$S'\coloneqq S\setminus\{v\}\cup\{w\}$ we have~$\val(S')>\val(S)$.
	Observe that~$\val(S\setminus\{v\})\ge\val(S)-1$.
	Furthermore, note that adding~$w$ may result into at most~$k$ new inner edges and hence the value decreases by at most~$k$.
	Simultaneously, since~$\deg(w)\ge \ell$, at least~$\ell-k$ new outer edges emerge such that the value increases by at least~$\alpha(\ell-k)$. 
	Since~$\ell>(\alpha k+k+1)\cdot\alpha^{-1}$, we obtain~$\val(S')>\val(S)$.

	Thus, in the following we can assume that~$S\cap L=\emptyset$.
	Let~$C^*\coloneqq C\cap S$,~$|C^*|=z$, and~$D^*\coloneqq S\setminus C^*\subseteq D$.
	In the following, we show that~$z=tq-q$ and that there exists an~$x\in[t]^q$ such that~$N(V_x)\cap C^*=\emptyset$.
	For this, we consider the cases that~$z<tq-q$ and that~$z>tq-q$.
	In both cases we verify that~$\val(S)<t$ for each solution with exactly~$z$ vertices in~$C$.

	\paragraph{Case~$1$:~$z\le tq-q-1$.} 
	In other words,~$z=tq-q-p$ for some~$p\in[tq-q]$.
	Recall that~$\deg(v)=\ell$ for each vertex~$v\in D$.
	Thus, by \Cref{lem-alpha-inner-or-out},~$\val(S)$ is maximized if~$m_H(D^*)+m_H(C^*,D^*)$ is minimized.
	Since~$z=|C^*|<tq-q$, it is possible that no vertex of~$D^*$ is adjacent to any vertex in~$C^*$. 
	Thus,~$\val(S)$ is maximized if~$D^*$ is an independent set and if~$E_H(C^*,D^*)=\emptyset$.
	Hence, 
	\begin{align*}
	\val(S)&\le (1-\alpha)\frac{(tq-q-p)(tq-q-p-1)}{2}\\
	&\hspace{1cm} +\alpha\left[(k+p)\ell+(tq-q-p)\left(\ell+\Bigl\lceil\frac{(3\alpha-1)tq}{\alpha}\Bigr\rceil\right)\right]\\
	&\hspace{1cm} -\alpha(tq-q-p)(tq-q-p-1)\eqcolon f(p).
	\end{align*}
	
	Now, we obtain that the derivative~$f'$ of~$f$ is

	\begin{align*}
	f'(p)=(1-3\alpha)p+\frac{3\alpha-1}{2}(2tq-2q-1)-\alpha\Bigl\lceil\frac{(3\alpha-1)tq}{\alpha}\Bigr\rceil.
	\end{align*}

	Since~$(1-3\alpha)<0$ for~$\alpha\in(1/3,1]$, we obtain for~$p\le tq-q$ that
	\begin{align*}
	f'(p)&\le\frac{3\alpha-1}{2}(2tq-2q-1)-\alpha\Bigl\lceil\frac{(3\alpha-1)tq}{\alpha}\Bigr\rceil \\
	&\le \frac{3\alpha-1}{2}(2tq-2q-1)-(3\alpha-1)tq=(1-3\alpha)(q+1/2)<0.
	\end{align*}

	Since~$\alpha\in(1/3,1]$ we conclude that~$f(p)$ is a concave quadratic function.
	And since~$f'(p)<0$ for each~$p\le tq-q$, we thus conclude that~$f(0)>f(p)$ for each~$p\in[tq-q]$, a contradiction to the assumption that~$f(p)=\val(S)\ge t'$.

%

	\paragraph{Case~$2$:~$z\ge tq-q+1$.}
	Let~$z=tq-q+p$ for some~$p\in[q]$.
	By the pigeonhole principle there exist at least~$p$ indices~$i\in[q]$ such that~$w^i_j\in S$ for each~$j\in[t]$. 
	Recall that by construction, each vertex~$v\in D$ has exactly one neighbor in the set~$\{w^i_j, j\in[t]\}$.
	Since~$|D^*|=k-p$ we conclude that~$m_H(C^*,D^*)\ge (k-p)p$.
	Recall that~$\deg(v)=\ell$ for each vertex~$v\in D$.
	Thus, by \Cref{lem-alpha-inner-or-out},~$\val(S)$ is maximal if~$m_H(D^*)+m_H(C^*,D^*)$ is minimal.
Hence,~$D^*$ is an independent set and~$m_H(C^*,D^*)=(k-p)p$.
	Thus,  
	\begin{align*}
	\val(S)&\le (1-\alpha)\left[(k-p)p+\frac{(tq-q+p)(tq-q+p-1)}{2}\right]+\alpha(k-p)(\ell-p)   \\ 
	&\hspace{1cm} +\alpha p\left[\ell+\Bigl\lceil\frac{(3\alpha-1)tq}{\alpha}\Bigr\rceil-tq+q-k+1\right]\\ 
	&\hspace{1cm} +\alpha(t-1)q\left[\ell+\Bigl\lceil\frac{(3\alpha-1)tq}{\alpha}\Bigr\rceil-tq+q-p+1\right]\eqcolon f(p).
	\end{align*}

	Now, we obtain that the derivative~$f'$ of~$f$ is

	\begin{align*}
	f'(p)=(3\alpha-1)p+\Bigl(1-3\alpha\Bigr)\Bigl(k+tq-q-\frac{1}{2}\Bigr)+\alpha\Bigl\lceil\frac{(3\alpha-1)tq}{\alpha}\Bigr\rceil.
	\end{align*}

	Since~$\alpha\in(1/3,1]$ and since~$p\in[q]$, we obtain that

	\begin{align*}
	f'(p)&\le(3\alpha-1)q+(1-3\alpha)(k+tq-q-1/2)+\alpha\Big(\frac{(3\alpha-1)tq}{\alpha}+1\Big)\\
	&\le(3\alpha-1)\Big(2q+1/2+\frac{\alpha}{3\alpha-1}-k\Big).
	\end{align*}

	Since~$k>3q+\lceil\frac{\alpha}{3\alpha-1}\rceil$ and since~$\alpha\in(1/3,1]$, we obtain that~$f'(p)<0$ for~$p\le q$.
	Furthermore, since~$\alpha\in(1/3,1]$ we see that~$f(p)$ is a convex quadratic function.
	Thus, we conclude that~$f(0)>f(p)$ for each~$p\in[q]$, a contradiction to the assumption that~$f(p)=\val(S)\ge t'$.

%
%
%

	Hence,~$|C^*|=tq-q$ and thus~$|D^*|=k$.
	According to \Cref{lem-alpha-inner-or-out},~$\val(S)$ is maximal if~$m_H(D^*)+m_H(C^*,D^*)$ is minimal.
	Observe that if~$D^*$ is an independent set and if~$E_H(C^*,D^*)=\emptyset$, then~$\val(S)=t'$.
	Otherwise, if~$E_H(C^*,D^*)\ne\emptyset$ or if~$D^*$ is no independent set, then~$\val(S)<t'$.
	Thus,~$E_H(C^*,D^*)=\emptyset$ and~$D^*$ is an independent set.
	Now, if there exist two vertices~$u,v\in D^*$ such that~$u\in V_x$ and~$v\in V_{y}$ with~$x\ne y$ for~$x,y\in [t]^q$, then~$(N(u)\cup N(v))\cap C\ge q+1$.
	Since~$|C^*|=tq-q$ this implies that~$E_H(C^*,D^*)\ne\emptyset$, a contradiction.
	Hence,~$D^*\subseteq V_x$ for some~$x\in[t]^q$.
Since~$D^*$ is an independent set, we conclude that the instance~$x$ contains an independent set of size at least~$k$.

	Hence, we have a weak-$q$-composition from \textsc{Independent Set} to \probBorMax{} in~$(q+3)$-closed graphs.
	Now, the proposition follows by \Cref{lemma:is:noq}.
\end{proof}
}

\begin{proposition}
\label{thm-lb-min-c-0-1/3}
	For each~$\alpha\in (0,1/3)$, \probMin{} does not admit a kernel of size~$\Oh(k^{c-3-\epsilon})$ unless \PHC{}.
\end{proposition}

{
\begin{proof}
The proof is similar to the proof of \Cref{thm-lb-max-c-1/3-1}: 
The main difference is that we now cannot add leaf vertices to ensure that all vertices in the instance gadgets (and also all vertices in the instance choice gadget) have the same degree since then an optimal solution would simply pick these leaf vertices.
Instead, we add large cliques to ensure that all vertices in the instance gadgets (and also all vertices in the instance choice gadget) have the same degree.
  We provide a weak~$q$-composition from \textsc{Independent Set} on 2-closed graphs to \probMin{} in~$(q+3)$-closed graphs.
Here, we assume that~$k>3q+\lceil\frac{\alpha}{1-3\alpha}\rceil$.
  
  \problemdef{Independent Set}
{A graph~$G=(V,E)$ and an integer~$k$.}
{Is there an independent set of size exactly~$k$?}

Let~$[t]^q$ be the set of $q$-dimensional vectors whose entries are in~$[t]$.
For a vector~$x\in[t]^q$ we denote by~$x_i$ the~$i$-th entry of~$x$.
Assume that~$q\ge 2$ is a constant and that we are given exactly~$[t]^q$ instances~$\mathcal{I}_{x} \coloneqq (G_{x}, k)$ of \textsc{Independent Set} on~$2$-closed graphs. 
Let $V_x \coloneqq V(G_x)$ for each $x \in [t]^q$ and let $D \coloneqq \bigcup_{x \in [t]^q} V_x$.
We construct an equivalent instance~$(H,k',t')$ of \probMin{} as follows.

\paragraph{Construction:} 
First, for each~$x\in[t]^q$ we add the graph~$G_x$ to~$H$.
In other words, we added the \emph{instance gadgets} to~$H$.
We then add a clique~$C$ (the \emph{instance choice gadget}) consisting of~$tq$ vertices to~$H$.
The vertices of~$C$ are denoted by~$w^i_j$ with~$i\in[q]$ and~$j\in[t]$.
Furthermore, for each~$x\in[t]^q$ and $v\in V_x$, we add the edge~$vw^i_{x_i}$ for each~$i\in[q]$.
Fix an integer~$\ell\ge t^q\cdot n+ \lceil(1-3\alpha)tq\cdot\alpha^{-1}\rceil$.
We add leaf vertices so that $\deg(w^i_j)=\ell-\lceil(1-3\alpha)tq\cdot\alpha^{-1}\rceil$ for each vertex~$w^i_j$ and~$\deg(v)=\ell$ for each vertex $v \in D$.
Next, for each added leaf vertex~$v$, we add a clique~$C_v$ of size~$\lceil 3t'\cdot\alpha^{-1}\rceil$ to~$H$.
By~$L$ we denote the set of newly added leaf vertices and vertices in the cliques~$C_v$ for the new leaf vertices.
Finally, we set~$k'\coloneqq k+tq-q$ and 
\begin{align*}
t'&\coloneqq (1-\alpha)\frac{(tq-q)(tq-q-1)}{2}\\ 
&\hspace{1cm} +\alpha\left[k\ell+(tq-q)\left(\ell-\Bigl\lceil\frac{(1-3\alpha)tq}{\alpha}\Bigr\rceil\right)-(tq-q)(tq-q-1)\right].
\end{align*}

\paragraph{Closure:} We show that~$H$ is~$(q + 3)$-closed.
Note that~$H[L]$ is a cluster graph, that is, a graph in which each connected component is a clique, and thus has closure one.
Furthermore, observe that each vertex in~$L$ has at most one neighbor in~$V(H)\setminus L$.
From now on the argumentation is completely analogous to the argumentation in \Cref{thm-lb-max-c-1/3-1}.
We thus omit it.


\paragraph{Correctness:} In the following, we prove that there exists an independent set of size exactly~$k$ for some instance~$\mathcal{I}_x$ with~$x\in [t]^q$ if and only if there exists a vertex set~$S$ of size exactly~$k'$ in~$H$ such that~$\val(S) \le t'$.

Suppose that instance~$\mathcal{I}_x$ has an independent set~$I$ of size exactly~$k$ for some~$x\in[t]^q$.
By~$C^*\coloneqq C\setminus \{w^i_{x_i} \mid i \in [q] \}$ we denote the non-neighbors of~$V_x$ in~$C$.
Note that~$|C^*|=(t-1)q$.
We show that~$S\coloneqq I\cup C^*$ is a solution of~$(H,k',t')$.
Clearly,~$|S|=k+tq-q=k'$.
Since no vertex of~$I$ is connected with any vertex in~$C^*$,~$I$ is an independent set, and since~$C^*$ is a clique of size~$(t-1)q$, we conclude that~$m_H(S)=(tq-q)(tq-q-1)/2$.
Furthermore, since each vertex in~$I$ has degree~$\ell$, and since each vertex in~$C^*$ has degree~$\ell-\lceil(1-3\alpha)tq\cdot\alpha^{-1}\rceil$, we conclude that~$m_H(S, V(H) \setminus S)=k\ell+(tq-q)(\ell-\lceil(1-3\alpha)tq\cdot\alpha^{-1}\rceil)-(tq-q)(tq-q-1)$.
Thus,~$\val(S)=t'$ and hence~$(H,k',t')$ is a yes-instance of \probMin{}.

Conversely, suppose that~$(H,k',t')$ has a solution~$S\subseteq V(H)$ of size exactly~$k'$ with~$\val(S)\le t'$. 
First, we show that we can assume that~$S\cap L=\emptyset$.
Assume that there exists a vertex~$v\in S\cap L$.
Since~$\deg(v)= \lceil 3t'\cdot\alpha^{-1}\rceil$ we conclude that~$m_H(\{v\},S\setminus\{v\})\ge 2t'\alpha$ and since~$\val(S)\ge \alpha\cdot m_H(\{v\},S\setminus\{v\})$ we obtain that~$\val(S)>t'$, a contradiction.

Thus, in the following we can assume that~$S\cap L=\emptyset$.
Let~$C^*\coloneqq C\cap S$,~$|C^*|=z$, and~$D^*\coloneqq S\setminus C^*\subseteq D$.
Next, we show that~$z=tq-q$ and that there exists an~$x\in[t]^q$ such that~$N(V_x)\cap C^*=\emptyset$.
To show that, we consider the cases that~$z<tq-q$ and that~$z>tq-q$.
In both cases we verify that~$\val(S)>t$ for each solution with exactly~$z$ vertices in~$C$.

\paragraph{Case~$1$:~$z\le tq-q-1$.} 
In other words,~$z=(t-1)q-p$ for some~$p\in[tq-q]$.
Recall that~$\deg(v)=\ell$ for each vertex~$v\in D$.
Thus, by \Cref{lem-alpha-inner-or-out},~$\val(S)$ is minimized if~$m_H(D^*)+m_H(C^*,D^*)$ is minimized. 
Since~$z=|C^*|<tq-q$, it is possible that no vertex of~$D^*$ is adjacent to any vertex in~$C^*$. 
Thus,~$\val(S)$ is minimized if~$D^*$ is an independent set and if~$E_H(C^*,D^*)=\emptyset$.
Hence, 
\begin{align*}
\val(S))&\ge (1-\alpha)\frac{(tq-q-p)(tq-q-p-1)}{2}\\
&\hspace{1cm} +\alpha\left[(k+p)\ell+(tq-q-p)\left(\ell-\Bigl\lceil\frac{(1-3\alpha)tq}{\alpha}\Bigr\rceil\right)\right]\\
&\hspace{1cm} -\alpha(tq-q-p)(tq-q-p-1)\eqcolon f(p).
\end{align*}

Now, we obtain that the derivative~$f'$ of~$f$ is

\begin{align*}
f'(p)=(1-3\alpha)p+\frac{3\alpha-1}{2}\Bigl(2tq-2q-1\Bigr)+\alpha\Bigl\lceil\frac{(1-3\alpha)tq}{\alpha}\Bigr\rceil.
\end{align*}

Since~$\alpha\in(0,1/3)$ and since~$p\le tq-q$, we obtain that

\begin{align*}
f'(p)\ge \frac{3\alpha-1}{2}\Bigl(2tq-2q-1\Bigr)+\alpha\frac{(1-3\alpha)tq}{\alpha}=(1-3\alpha)(q+1/2)>0.
\end{align*}

Since~$\alpha\in[0,1/3]$ we conclude that~$f(p)$ is a convex quadratic function.
And since~$f'(p)>0$ for each~$p\le q$, we thus conclude that~$f(0)<f(p)$ for each~$p\in[q]$, a contradiction to the assumption that~$f(p)=\val(S)\le t'$.

%

\paragraph{Case~$2$:~$z\ge tq-q+1$.}
Let~$z=tq-q+p$ for some~$p\in[q]$.
By the pigeonhole principle there exist at least~$p$ indices~$i\in[q]$ such that~$w^i_j\in S$ for each~$j\in[t]$. 
Recall that by construction, each vertex~$v\in D$ has exactly one neighbor in the set~$\{w^i_j, j\in[t]\}$.
Since~$|D^*|=k-p$ we conclude that~$m_H(C^*,D^*)\ge (k-p)p$.
Recall that~$\deg(v)=\ell$ for each vertex~$v\in D$.
Thus, by \Cref{lem-alpha-inner-or-out},~$\val(S)$ is minimal if~$m_H(D^*)+m_H(C^*,D^*)$ is minimal.
Hence,~$D^*$ is an independent set and~$m_H(C^*,D^*)=(k-p)p$.
Thus,  
\begin{align*}
\val(S)&\ge (1-\alpha)\left[(k-p)p+\frac{(tq-q+p)(tq-q+p-1)}{2}\right]+\alpha(k-p)(\ell-p)  \\ 
&\hspace{1cm} +\alpha p\left[\ell-\Bigl\lceil\frac{(1-3\alpha)tq}{\alpha}\Bigr\rceil-tq+q-k+1\right]\\
&\hspace{1cm} + \alpha(tq-q)\left[\ell-\Bigl\lceil\frac{(1-3\alpha)tq}{\alpha}\Bigr\rceil-tq+q-p+1\right]\eqcolon f(p).
\end{align*}

Now, we obtain that the derivative~$f'$ of~$f$ is

\begin{align*}
f'(p)=(3\alpha-1)p+\Bigl(1-3\alpha\Bigr)\Bigl(k+tq-q-\frac{1}{2}\Bigr)-\alpha\Bigl\lceil\frac{(1-3\alpha)tq}{\alpha}\Bigr\rceil.
\end{align*}

From~$p\le q$ and~$3\alpha-1<0$ for~$\alpha\in(0,1/3)$, we obtain that

\begin{align*}
f'(p)&\ge (3\alpha-1)q+(1-3\alpha)\Big(k+tq-q-1/2\Big)-\alpha\Big(\frac{(1-3\alpha)tq}{\alpha}+1\Big) \\
&=(1-3\alpha)\Big(k-2q-1/2-\frac{\alpha}{1-3\alpha}\Big).
\end{align*}

Since~$k>3q+\lceil\frac{\alpha}{1-3\alpha}\rceil$ and since~$\alpha\in(0,1/3)$, we obtain that~$f'(p)>0$ for~$p\le q$.
Furthermore, since~$\alpha\in(0,1/3)$ we see that~$f(p)$ is a concave quadratic function.
Thus, we conclude that~$f(0)<f(p)$ for each~$p\in[q]$, a contradiction to the assumption that~$f(p)=\val(S)\le t'$.

%
%
%

Hence,~$|C^*|=tq-q$ and thus~$|D^*|=k$.
According to \Cref{lem-alpha-inner-or-out},~$\val(S)$ is minimal if~$m_H(D^*)+m_H(C^*,D^*)$ is minimal.
Observe that if~$D^*$ is an independent set and if~$E_H(C^*,D^*)=\emptyset$, then~$\val(S)=t'$.
Otherwise, if~$E_H(C^*,D^*)\ne\emptyset$ or if~$D^*$ is no independent set, then~$\val(S)>t'$.
Thus,~$E_H(C^*,D^*)=\emptyset$ and~$D^*$ is an independent set.
Now, if there exist two vertices~$u,v\in D^*$ such that~$u\in V_x$ and~$v\in V_{y}$ with~$x\ne y$ for~$x,y\in [t]^q$, then~$(N(u)\cup N(v))\cap C\ge q+1$.
Since~$|C^*|=tq-q$ this implies that~$E_H(C^*,D^*)\ne\emptyset$, a contradiction.
Hence,~$D^*\subseteq V_x$ for some~$x\in[t]^q$.
Since~$D^*$ is an independent set, we conclude that the instance~$x$ contains an independent set of size at least~$k$.

Hence, we have a weak-$q$-composition from \textsc{Independent Set} to \probMin{} in~$(q+3)$-closed graphs.
Now, the proposition follows by \Cref{lemma:is:noq}.
\end{proof}
}

Recall that \probMin{} for~$\alpha=0$ is equivalent to \textsc{Sparest~$k$-Subgraph} which admits a kernel of size~$\Oh(c^2k^3)$~\cite{KKS20a}.

Now, \Cref{thm-lb-max-c-1/3-1,thm-lb-min-c-0-1/3,thm:closure:kernel} imply the following.

\begin{theorem}
	\label{thm:degrading:closure}
	\probBor{} admits a kernel of size~$k^{\Oh(c)}$.
	For~$\alpha > 0$, \probBor{} does not admit a kernel of size~$k^{o(c)}$ unless \PHC{}.
\end{theorem}

 \subsection{Hardness for the Non-Degrading Case}
In contrast to the degrading case, we show that the non-degrading case is intractable.

\begin{theorem}
	\label{thm-hardness-small-alpha-for-deg-and-closure}
	\probIntMax{} remains W[1]-hard with respect to the solution size~$k$ even on 2-closed and 2-degenerate graphs. 
\end{theorem}

{
\begin{proof}
	We reduce from the W[1]-hard \textsc{Densest-$k$-Subgraph} problem in $2$-degenerate and 2-closed graphs~\cite{RS08}. 
	Recall that \textsc{Densest~$k$-Subgraph} is the special case of \probMax{} with~$\alpha=0$.
	Hence, it remains to show the theorem for~$\alpha\in(0,1/3)$.
	Let~$(G,k,t)$ be an instance of \textsc{Densest~$k$-Subgraph}. 
	We construct an equivalent instance $(G',k',t')$ of \probMax{} as follows: 
	Initially, graph~$G'$ consists of a copy of~$G$. 
	Let~$Z$ denote the set of all vertices which are a copy of a vertex in~$G$. 
	We add exactly~$n(G)-\deg_G(v)$ many leaf-vertices to each vertex~$v\in Z$.
	We denote these vertices by~$I_v$.
	By~$I\coloneqq\bigcup_{v\in V(G)} I_v$ we denote the set of all these leaf-vertices. 
	Finally, we set~$k'\coloneqq k$ and~$t'\coloneqq\alpha(kn(G)-2t)+(1-\alpha)t$.

	By construction, each vertex in~$I$ has degree~$1$.
	Since~$G$ is~$2$-closed and~$2$-degenerate, we conclude that also~$G'$ is~$2$-closed and~$2$-degenerate.

	Let~$S\subseteq V(G)$ be such that~$|S|= k$ and that~$m(G[S])\ge t$.  
	Since
\begin{itemize}
\item each vertex~$v\in V(G')\cap Z$ has degree~$n(G)$,
\item $N(x)\cap N(y)\subseteq Z$ for each two vertices~$x,y\in Z$, and
\item $m(G[S])\ge t$,
\end{itemize}	
 we conclude that exactly~$x\ge t$ edges have both endpoints in~$S$ and that exactly~$kn(G)-2x\le kn(G)-2t$ edges have exactly one endpoint in~$S$.
	Thus,~$\val(S)\ge\alpha(kn(G)-2t)+(1-\alpha)t=t'$ since~$\alpha\in(0,1/3)$.

	Conversely, suppose that~$S' \subseteq V(G')$ is set of exactly~$k$ vertices with~$\val(S')\ge t'$. 
	By construction, each vertex in~$I$ has degree~$1$. 
	Hence,~$\contribution(v)\le \max(\alpha,1-\alpha)<1$ for each vertex~$v\in I$.
	Furthermore, observe that for each vertex~$z\in Z$ we have~$\contribution(w)\ge \alpha n(G)>1$ since~$\alpha$ is a constant.
	Thus,~$\contribution(z)>\contribution(v)$ for each vertex~$z\in Z$ and each~$v\in I$.
	Hence, we can assume that~$S'\subseteq Z$.
	Let~$x$ be the number of edges with both endpoints in~$S'$.
	Since each vertex in~$S'$ has degree~$n(G)$, we conclude that exactly~$kn(G)-2x$ edges have exactly one endpoint in~$S'$.
	Hence, $\val(S')=\alpha(kn(G)-2x)+(1-\alpha)x$.
	Since~$\val(S')\ge t'$ and~$\alpha\in(0,1/3)$, we conclude that~$x\ge t$.
	Hence, $G[S']$ contains at least~$t$ edges.
\end{proof}
}


\begin{theorem}
	\label{thm:nondeg:max:c}
	\probIntMin{} remains W[1]-hard with respect to the solution size~$k$ even on 2-closed graphs. 
\end{theorem}

{
\begin{proof}
	We reduce from the W[1]-hard \textsc{Densest-$k$-Subgraph} problem in 2-closed graphs~\cite{RS08}. 
	Let~$(G,k,t)$ be an instance of \textsc{Densest~$k$-Subgraph}. 
	We construct an equivalent instance $(G',k',t')$ of \probMin{} as follows: 
	Initially, graph~$G'$ consists of a copy of~$G$. 
	Let~$Z$ denote the set of all vertices which are a copy of a vertex in~$G$. 
	For each vertex~$v\in Z$ we add a set~$W_v$ of exactly~$n(G)-\deg_G(v)$ many vertices to~$G'$ such that each vertex of~$W_v$ is adjacent with vertex~$v$. 
	By~$W\coloneqq\bigcup_{v\in V(G)}W_v$ we denote the set of all these vertices. 
	Next, for each vertex~$w\in W$ we add a clique~$C_w$ of size~$4t'$ to~$G'$ such that each vertex of~$C_w$ is adjacent with~$w$. 
	By~$C\coloneqq\bigcup_{w\in W}C_w$ we denote the set of all these vertices. 
	Finally, we set~$k'\coloneqq k$ and~$t'\coloneqq\alpha(kn(G)-2t)+(1-\alpha)t$.

	Observe that~$G'[W\cup C]$ is a cluster graph, that is, a disjoint union of cliques. 
	Furthermore, observe that each vertex in~$W\cup C$ has at most one neighbor in~$Z$. 
	Since~$G[Z]$ is 2-closed, we conclude that~$G'$ is 2-closed.

	Let~$S\subseteq V(G)$ be such that~$|S|= k$ and that~$m(G[S])\ge t$.  
	Since
	\begin{itemize}
	\item each vertex~$v\in V(G')\cap Z$ has degree~$n(G)$,
	\item $N(x)\cap N(y)\subseteq Z$ for each two vertices~$x,y\in Z$, and
	\item $m(G[S])\ge t$,
	\end{itemize}
we conclude that exactly~$x\ge t$ edges have both endpoints in~$S$ and that exactly~$kn(G)-2x\le kn(G)-2t$ edges have exactly one endpoint in~$S$.	
	Thus, $\val(S)\le\alpha(kn(G)-2t)+(1-\alpha)t=t'$ since~$\alpha\in(1/3,1]$.

	Conversely, suppose that~$S' \subseteq V(G')$ is set of exactly~$k$ vertices with~$\val(S') \le t'$. 
	By construction, each vertex in~$W\cup C$ has degree at least~$2t'$. 
	If~$S'$ contains a vertex~$v$ of~$W\cup C$, then~$v$ has at least~$3t'$ neighbors which are not in~$S'$ since~$t'>k$. 
	Thus, $\val(S')> \alpha 3t'>t'$, a contradiction to the assumption~$\val(S')\le t'$.
	Hence, $S'\cap (W\cup C)=\emptyset$. 
	Let~$x$ be the number of edges with both endpoints in~$S'$.
	Since each vertex in~$S'$ has degree~$n(G)$, we conclude that exactly~$kn(G)-2x$ edges have exactly one endpoint in~$S'$.
	Hence, $\val(S')=\alpha(kn(G)-2x)+(1-\alpha)x$.
	Since~$\val(S')\le t$ and~$\alpha\in(1/3,1]$, we conclude that~$x\ge t$.
	Hence, $G[S']$ contains at least~$t$ edges.
\end{proof}
}

\section{Parameterization by Degeneracy}\label{sec:degeneracy}

We show that in the minimization variant we obtain an FPT-algorithm for each~$\alpha$.
For minimization in the degrading variant we even obtain polynomial kernels for~$d+k$.
In contrast, for maximization in the degrading variant we  provide a tight kernel of size~$k^{\Oh(d)}$.
To prove this result we again rely on a Ramsey bound.

 \subsection{Minimization Variant}
We start with the minimization variant which turns out to be easier than the maximization variant. 
This is most likely because of the following bound on~$t$.

\begin{lemma}
	\label{lem-minimization-delta-bounded-by-d+k}
	Let~$(G,k,t)$ be an instance of \probMin{}. 
	If~$t\ge dk$, then~$(G,k,t)$ is a trivial yes-instance..
\end{lemma}

{
\begin{proof}
If~$|V(G)|<k$, then we have a trivial no-instance and the statement follows immediately.
Hence, in the following we assume that~$|V(G)|\ge k$.
Next, we show that there exists a solution~$S$ such that~$|F|\le dk$ where~$F$ is the set of edges with at least one endpoint in~$S$.
Then, it follows that~$\val(S)\le t$ in such instances of \probMin{} and the statement is proven.
Let~$\sigma$ be a degeneracy ordering of~$G$ and let~$S\coloneqq\{v_1,
\ldots, v_k\}$ consist of the first~$k$ vertices of~$\sigma$.
Each edge in~$F$ is of the form~$xy$ where~$x<y$ with respect to~$\sigma$.
Since~$S$ consists of the first~$k$ vertices of~$\sigma$ and~$G$ is~$d$-degenerate, we conclude that for~$x=v_i$ with~$i\in[k]$ there are at most~$d$ edges in~$F$ of the form~$xy$.
Since~$|S|=k$,~$F$ contains at most~$dk$ edges.
\end{proof}
}

Shachnai and Zehavi~\cite{SZ17} showed that \probMin{} with~$\alpha\in(0,1]$ admits an FPT-algorithm with respect to~$k+t$.
Hence, we obtain the following.

\begin{corollary}
\label{lem-minimization-fpt-for-d}
\probMin{} for~$\alpha>0$ is FPT parameterized by~$d+k$. 
\end{corollary}

Naturally, we may now ask whether this FPT result can be strengthened to a polynomial kernel.
As shown by~\Cref{thm:maxdeg:no-poly-kernel-const-delta}, the \internalFocusedS{} case of \probMin{} does not admit a polynomial kernel even on graphs with constant maximum degree which implies constant degeneracy. 
In contrast, the \borderFocusedS{} variant has a kernel whose size is polynomial in~$d+k$.

\begin{theorem}
	\label{thm-minimization-kernel-d+k}
	\probBorMin{} admits a kernel of size~$(d+k)^{\Oh(1)}$.
\end{theorem}

{
First, we consider the case~$\alpha=0$.
Recall that \textsc{Sparsest~$k$-Subgraph} is the special case of \prob{} for minimization and~$\alpha=0$.}

\begin{proposition}
\label{prop-minimization-kernel-d+k-alpha=0}
{\normalfont\textsc{Sparsest~$k$-Subgraph}} admits a kernel with $\Oh(dk)$~vertices and of size~$\Oh(d^2k)$.
\end{proposition}

{
\begin{proof}
Let~$(G,k,t)$ be an instance of \textsc{Sparsest~$k$-Subgraph}.
Assume that~$|V(G)|\ge dk$.
Now, since~$G$ is~$d$-degenerate, we obtain an independent set~$I$ of size at least~$k$ in~$G$.
Observe that~$\val(S)=0$ and thus~$(G,k,t)$ is a trivial yes-instance.
Hence,~$|V(G)|<dk$.
Since the number of edges in a~$d$-degenerate graph is bounded by~$d\cdot|V(G)|$ the statement follows.
\end{proof}

Now, we consider~$\alpha\in(0,1/3)$.

\begin{lemma}
\label{lem-minimization-kernel-d+k-alpha>0}
\probMin{} for~$\alpha\in(0,1/3)$ admits a kernel of size~$\Oh(d^4k^5)$.
\end{lemma}
\begin{proof}
Let~$I$ be an instance of \probMin{} for a fixed~$\alpha\in(0,1/3)$.
First, we transform~$I$ into an equivalent instance~$I'\coloneqq (G,\emptyset,\counter,k,t)$ of \probAnoMin{} with parameter~$\alpha$.
Note that the upper bound for~$t$ from \Cref{lem-minimization-fpt-for-d} still holds for~$I'$.
Observe that if~$G$ contains a vertex~$v$ such that~$\alpha\degCounter(v)\ge t+k$, then we have for each solution~$S$ containing~$v$ that~$\val(S)>t$.
Hence, if there is a solution~$S$ for~$I'$, then~$v\notin I$.
Thus, it is safe to apply the Exclusion Rule (\Cref{rr:general:how-to-exclude}) to vertex~$v$.
In the following, we assume that the Exclusion Rule (\Cref{rr:general:how-to-exclude}) is applied exhaustively.

Note that now we have~$\Delta(G)<(t+k)\cdot\alpha^{-1}$ and also~$\Gamma(G)<(t+k)\cdot\alpha^{-1}$.
Next, we can apply \Cref{rr:delta:better} exhaustively.
Analogously to the proof of \Cref{prop:general:delta}, after the exhaustive application of \Cref{rr:delta:better} we have~$|V(G)|\le\Delta k$.
Afterwards, with the reduction described in \Cref{lemma:general:remove-ano-min}, we can construct an equivalent instance~$(G',k,t'')$ of \probMin{} of size~$\Oh((\Delta(G)+\Gamma(G))^3|V(G)|)=\Oh(d^4k^5)$.
\end{proof}

Now, \Cref{prop-minimization-kernel-d+k-alpha=0} and \Cref{lem-minimization-kernel-d+k-alpha>0} lead to \Cref{thm-minimization-kernel-d+k}.
}

 \subsection{Maximization Variant}
Recall that \textsc{MaxPVC} is the special case of \probMax{} with~$\alpha=1/2$.
Amini et al.~\cite{AFS11} showed that \textsc{MaxPVC} can be solved in $\Oh^*((dk)^k)$~time. 
Adapting this algorithm leads to an FPT-algorithm for \prob{} with respect to~$d+k$ for $\alpha \ne 0$.
The main distinction is that in our adaptation it is important to use the annotated variant to keep track of vertices being contained in a partial solution.

\begin{proposition}
	\label{prop-fpt-d+k-degrading}
	 \probBor{} can be solved in $\Oh^*((dk)^k)$~time for~$\alpha \ne 0$.
\end{proposition}

\begin{proof}
Observe that it is sufficient to present an algorithm with running time~$\Oh^*((dk)^k)$ for \probAnoBor{} for~$\alpha\ne 0$ since each instance~$(G,k,T)$ of \probBor{} can be transformed into an equivalent instance~$(G,\emptyset,\counter,k,t)$ of \probAno{}.
We present a search-tree algorithm with depth at most~$k$ such that each node has at most $dk$~children.
The root of the search-tree corresponds to the instance~$(G,\emptyset,\counter,k,t)$.

Now, we consider the instance~$I\coloneq (G,T,\counter,k',t')$ instance of \probAnoBor{} with~$\alpha\ne 0$.
Let~$L\coloneqq \{v\in V(G)\setminus T\mid \contribution(v,T)\ge t'/k'\}$ (maximization) and~$L\coloneqq \{v\in V(G)\setminus T\mid \contribution(v,T)\le t'/k'\}$ (minimization).
Clearly, if~$|L|\ge dk$, then there exists a subset~$I\subseteq L$ such that~$I$ induces an independent set in~$G$. 
For maximization we conclude from~$\contribution(v)\ge t'/k'$ for each vertex~$v\in I$ we that~$\val(I)\ge t'$.
Otherwise,~$|L|< dk$ and for minimization we conclude from~$\contribution(v)\le t'/k'$ for each vertex~$v\in I$ we that~$\val(I)\le t'$.
Otherwise,~$|L|< dk$.  
By definition of~$L$ we have~$\contribution(u)<t'/k'$ (maximization) and~$\contribution(u)>t'/k'$ (minimization) for each~$u\in V(G)\setminus (L\cup T)$.
Thus,~$k$ vertices from~$V(G)\setminus (L\cup T)$ are not sufficient to obtain value at least (maximization) or at most (minimization)~$t'$, 
Hence, each solution contains at least one vertex of~$L$.
In other words,~$I$ is a yes-instance if and only if at least one of the instance~$(G,T\cup\{v\},\counter,k'-1,t'-\contribution(v)$ for some~$v\in L$ is a yes-instance.
Hence, each node in the enumeration tree has at most $dk$~children.
Furthermore, since in each child of~$I$ the parameter~$k'$ is reduced by~$1$, the depth is bounded by~$k$.
Hence, we obtain an~$\Oh^*((dk)^k)$ for \probAno{}.
%
\end{proof}


The rest of this section is devoted to the proof of the next theorem.

\begin{theorem}
	\label{thm:degeneracy:max-border-case}
	\probBorMax{} admits a kernel of size~$k^{\Oh(d)}$ but, unless \PHC, no kernel of size $\Oh(k^{d-2-\epsilon})$.
\end{theorem}

In particular, this implies that \textsc{MaxPVC} admits a kernel of size $k^{\Oh(d)}$.
We remark that a compression of size $(dk)^{\Oh(d)}$ was obtained independently by Panolan and Yaghoubizade~\cite{PY22}.

\paragraph{A kernel for biclique-free graphs in the degrading case.}
\label{sec:bc-free}
We next develop a kernelization algorithm with the size bound $k^{\Oh(d)}$.
In fact, our algorithm works for \emph{biclique-free graphs}---graphs that do not have a biclique $K_{a, b}$ as a subgraph for $a \le b \in \mathds{N}$.
Note that a $d$-degenerate graph has no $K_{d + 1, d + 1}$ as a subgraph, since every vertex in a $K_{d + 1, d + 1}$ has degree $d + 1$.

Note that a clique of size $a + b$ contains $K_{a, b}$ as a subgraph.
So given a graph $G$ with no occurrence of $K_{a, b}$ on at least $\binom{a + b + k - 2}{k - 1} \in k^{\Oh(a + b)}$ vertices, one can find an independent set of size $k$ in polynomial time (see \Cref{sec:prelim}).
We show that this upper bound on the number of vertices can be improved: the sum $a + b$ in the exponent can be replaced by $\min \{ a, b \}$.

\begin{lemma}
  \label{lemma:bc-free:ramsey}
  For $a \le b \in \mathds{N}$, let $G$ be a graph that contains no $K_{a, b}$ as a subgraph.
  If $G$ has at least $R(k)$ vertices, then we can find in polynomial time an independent set of size $k$, where $R(k) \in (a + b)^{\Oh(a)} \cdot k^a$
\end{lemma}

\begin{proof}
  We first show that if $G$ has at least $k + b \binom{k}{a} + \sum_{\ell \in [a - 1]} R(a + b, \ell + 1) \binom{k}{\ell}$ vertices, then it contains an independent set of size $k$.
  Afterwards, we give a polynomial-time algorithm to find such an independent set of size $k$.
  Let $I$ be a maximum independent set in $G$.
  We assume for contradiction that $|I| < k$.
  We prove that there are at most $t \binom{k}{a}$ vertices that have at least $a$ neighbors in $I$ and that there are at most $\sum_{\ell \in [a - 1]} R(a + b, \ell + 1)$ vertices that have at most $a - 1$ neighbors in $I$.

  For each subset $X \subseteq I$ of size exactly $a$, note that there are at most $b$ vertices $v$ such that~$N(v) \supseteq X$, since otherwise there is a $K_{a, b}$ in $G$.
  It follows that the number of vertices with at least $a$ neighbors in $I$ is at most $t \binom{|I|}{a} \le b \binom{k}{a}$.
  Consider a set $X \subseteq I$ of size $\ell \in [a - 1]$.
  Let~$V_X \coloneqq \{ v \in V(G) \setminus I \mid N(v) \cap I = X \}$.
  Then, there is no independent set $I'$ of size $\ell + 1$ in~$V_X$, since otherwise $(I \setminus X) \cup I'$ is an independent set of size at least $|I| + 1$, contradicting the fact that $I$ is an independent set of maximum size.
  Moreover, there is no clique of size~$a + b$ in $V_X$.
  Thus, $|V_X| < R(a + b, \ell + 1)$.
  The number of vertices with at most $a - 1$ neighbors in~$I$ is then at most $\sum_{X \subseteq I, |X| = \ell \in [a - 1]} R(a + b, \ell + 1) \binom{|I|}{\ell} \le \sum_{\ell \in [a - 1]} R(a + b, \ell + 1) \binom{k}{\ell}$.

  We turn the argument above into a polynomial-time algorithm as follows.
  Suppose that we have an independent set $I'$ of size smaller than $k$.
  As discussed above, there are at most~$b \cdot \binom{k}{a}$ vertices that have at least $s$ neighbors in $I'$.
  Hence, there is a vertex set $X \subseteq I'$ of size $\ell$ such that $|V_X| > R(a + b, \ell + 1)$.
  Note that $X$ can be found in polynomial time, for instance, by counting the number of vertices $v'$ such that $N(v') \cap I' = N(v) \cap I'$ for each vertex $v \in V(G)$.
  We can then find an independent set $I''$ of size $\ell + 1$ in $X$ (this can be done in polynomial time as discussed in \Cref{sec:prelim}).
  This way, we end up with an independent set~$(I' \setminus X) \cup I''$ of size at least $|I'| + 1$.
  Note that this procedure of finding an independent set of greater size is repeated at most $k$ times, and thus the overall running time is polynomial.
\end{proof}

We remark that for fixed $a \le b \in \mathbb{N}$, \Cref{lemma:bc-free:ramsey} gives us an $\Oh(n^{1 - 1 / a})$-approximation algorithm for \textsc{Independent Set} that runs in $n^{\ell}$ time for some constant $\ell$ not depending on $a$ or $b$.
An $\Oh(n^{1 - 1 / a})$-approximation algorithm is known on graphs where $K_{a, b}$ is excluded as an \emph{induced subgraph} \cite{BTTW20,DFRR20}.
However, these algorithms have running time $n^{\Omega(a)}$.

We now apply \Cref{lemma:bc-free:ramsey} to obtain a lemma analogous to \Cref{lemma:pvc:findy}.

\begin{lemma}
  \label{lemma:bc-free:largedeg}
  Suppose that $\Delta \ge R(b k^{a - 1})$.
  Then, we can find in polynomial time a set $X$ of $i \in [a - 1]$ vertices and an independent set $I$ with the following properties:
  \begin{enumerate}
    \item
      The set $I \subseteq \bigcap_{x \in X} N(x)$ is an independent set of size at least $b k^{a - i} + 1$.
    \item
      For every vertex $u \in V(G) \setminus X$, it holds that $|N(u) \cap I| \le bk^{a - i - 1}$.
  \end{enumerate}
\end{lemma}

{
\begin{proof}
  Let $v$ be a vertex with $\deg(v) \ge R(b k^{a - 1})$.
  By \Cref{lemma:bc-free:ramsey}, there is an independent set $I_v$ of size $b k^{b - 1}$ in $N(v)$ (which can be found in polynomial time).
  Let $X$ be an inclusion-wise maximal set of $i$ vertices containing $v$ with $|\bigcap_{x \in X} N(x) \cap I_v| > bk^{a - i}$.
  Such a set can be found by the following polynomial-time algorithm:
  We start with $X = \{ v \}$ and $i = 1$.
  We will maintain the invariant that $|X| = i$.
  If there exists a vertex $v' \in V(G) \setminus X$ with $|N(v') \cap \bigcap_{x \in X} N(x) \cap I_v| > b k^{a - i - 1}$, then we add $v'$ to $X$ and increase $i$ by 1.
  We keep doing so until there remains no such vertex $v'$.

  We show that this algorithm terminates for $i = |X| \le a - 1$.
  Assume to the contrary that the algorithm continues for $i = a - 1$.
  We then have that $|N(v') \cap \bigcap_{x \in X} N(x) \cap I_v| > b k^{a - i - 1}$ for some vertex $v' \in V(G) \setminus X$.
  It follows that the set $X \cup \{ v' \}$ (which is of size $a$) has more than $b$ common neighbors, contradicting the fact that $G$ has no $K_{a, b}$ as a subgraph.

  Finally, we show that the set $X$ found by this algorithm and $I \coloneqq \bigcap_{x \in X} N(x) \cap I_v$ satisfy the three properties of the lemma.
  We have $|I| = |\bigcap_{x \in X} N(x) \cap I_v| = |N(v') \cap \bigcap_{x \in X \setminus \{ v' \}} N(x) \cap I_v| > b k^{a - (i - 1) - 1} = b k^{a - i}$, where $v'$ is the last vertex added to $X$.
  Moreover, since $X$ is inclusion-wise maximal, we have $|N(u) \cap I| = |N(u) \cap  \bigcap_{x \in X} N(x) \cap I_v| \le b k^{a - i - 1}$ for every vertex $u \in V(G) \setminus X$.
\end{proof}
}

\begin{rrule}
	\label{rr:bc-free:independent-set}
	Let~$\mathcal{I}$ be an instance of \probAnoBor{}.
	Let~$X, I$ be as specified in \Cref{lemma:bc-free:largedeg} and let $v \in I$ be a vertex such that every other vertex in $I$ is better than $v$.
	Then, apply the Exclusion Rule (\Cref{rr:general:how-to-exclude}) to~$v$.
\end{rrule}

{
\begin{lemma}
  \Cref{rr:bc-free:independent-set} is correct.
\end{lemma}
\begin{proof}
  We show the proof for the maximization variant; the minimization variant follows analogously.
  For the sake of contradiction, assume that every solution $S$ contains $v$.
  By \Cref{lemma:bc-free:largedeg}, every vertex $u \in V(G) \setminus X$ has at most $b k^{a - i}$ neighbors in $I$.
  Moreover, since $I$ is an independent set, we have $|I \cap N[v']| = 1$ for every vertex $v' \in I$ (including $v$).
  For $S' \coloneqq S \setminus X$, we have
  \begin{align*}
    |I \setminus N[S']|
    &\ge |I| - |I \cap N[v]| - |I \cap N[S' \setminus \{ v \}]| \\
    &\ge (b k^{a - i} + 1) - (k - 1) b k^{a - i - 1} - 1
    = b k^{a - i - 1} > 0.
  \end{align*}
  Let $v'$ be an arbitrary vertex in $I \setminus N[S']$.
  We show that $\contribution(v', S \setminus \{ v \}) \ge \contribution(v, S \setminus \{ v \})$.
  By \Cref{lemma:replace}, this would imply that $(S \setminus \{ v \}) \cup \{ v' \}$ is a solution not containing $v$.
  Since $v$ and $v'$ are both adjacent to all vertices of $X$ and $\alpha \in (1/ 3, 1]$, we have
  \begin{align*}
    \contribution(v', S \setminus \{ v \})
    &= \alpha \degCounter(v') + (1 - 3 \alpha) |X \cap (S \setminus \{ v \})| \\
    &\ge \alpha \degCounter(v) + (1 - 3 \alpha) |X \cap (S \setminus \{ v \})| \\ 
    &\ge \alpha \degCounter(v) + (1 - 3 \alpha) |N(v) \cap (S \setminus \{ v \})| = \contribution(v, S \setminus \{ v \}).
  \end{align*}
  Here, the first inequality follows from the fact that $v'$ is better than $v$.
\end{proof}
}

By applying \Cref{rr:bc-free:independent-set} exhaustively, we end up with an instance with maximum degree $\Delta \le R(b k^{a - 1})$.
The following proposition then follows from \Cref{prop:general:kernel} using the bound in \Cref{lemma:bc-free:ramsey}:

\begin{proposition}
  \label{thm:bc-free:kernel}
  For any $a \le b \in \mathbb{N}$, \probBor{} on graphs that do not contain~$K_{a, b}$ as a subgraph has a kernel of size $(R(b k^{a - 1}) + k)^{\Oh(1)} \in b^{\Oh(a)} k^{\Oh(a^2)}$.
\end{proposition}

Note that a $d$-degenerate graph contains no $K_{d + 1, d + 1}$ as a subgraph.
Thus, we obtain a kernel of size $k^{\Oh(d^2)}$ for fixed $d$.
In fact, we obtain a smaller kernel using the folklore fact that any $d$-degenerate graph on at least $(d + 1)k$ vertices has an independent set of size $k$.

\begin{lemma}
  \label{cor:degeneracy:kernel}
  \probBor{} admits a kernel of size $k^{\Oh(d)}$.
\end{lemma}

{
\begin{proof}
	Since a $d$-degenerate graph has no $K_{d + 1, d + 1}$ as a subgraph, we conclude that there is a kernel of size $(R(d k^{d - 1}) + k)^{O(1)}$ by \Cref{thm:bc-free:kernel}.
	Recall that for $\ell \in \mathds{N}$, $R(\ell)$ denotes an integer such that any graph on $R(\ell)$ vertices has an independent set of size $\ell$.
	Since $R(\ell) \le (d + 1) \ell$ for any $d$-degenerate graphs, there is a kernel of size $(R(d k^{d - 1}) + k)^{O(1)} \in k^{\Oh(d)}$.
\end{proof}
}

 \paragraph{A matching Lower Bound.} 
Now, we show that significant improvement in \Cref{cor:degeneracy:kernel} is unlikely.
This, together with \Cref{cor:degeneracy:kernel}, implies \Cref{thm:degeneracy:max-border-case}.

\begin{proposition}
	\label{thm:degenercy:nokernel}
	\probBorMax{} admits no kernel of size $\Oh(k^{d-2-\epsilon})$ unless \PHC{}.
\end{proposition}

\begin{proof}
	We provide a weak~$q$-composition from \textsc{Independent Set} on 2-degenerate graphs to \probBorMax{} in~$(q+2)$-degenerate graphs.
	Here, we assume that~$k>|c(\alpha)|/(3\alpha-1)$ for some constant~$c(\alpha)>0$ which will be specified below.
	

	Let~$[t]^q$ be the set of $q$-dimensional vectors whose entries are in~$[t]$.
	For a vector~$x\in[t]^q$ we denote by~$x_i$ the~$i$th entry of~$x$.
	Next, assume that~$q\ge 2$ is a constant and that we are given exactly~$t^q$ instances~$\mathcal{I}_{x} \coloneqq (G_{x}, k)$ of \textsc{Independent Set} on~$2$-degenerate graphs. 
	We construct an equivalent instance~$(H,k',t')$ of \probBorMax{} as follows.

	\paragraph{Construction:} We add an independent set~$J$ (the \emph{instance choice gadget}) consisting of~$tq$ vertices to~$H$.
	The vertices of~$J$ are denoted by~$w^i_j$ with~$i\in[q]$ and~$j\in[t]$.
	Next, for each~$x\in[t]^q$ we add the graph~$G_x$ to~$H$.
	In other words, we added the \emph{instance gadgets} to~$H$.
	By~$D$ we denote the union of these vertices.
	Furthermore, for each~$x\in[t]^q$ and for each vertex~$v\in V_x$ we add the edge~$vw^i_{x_i}$ for each~$i\in[q]$.
	Now, we fix an integer~$\ell> t^q\cdot n>k$.
	For each vertex~$w^i_j$ in the independent set~$D$ we add leaf vertices such that~$\deg(w^i_j)=\ell+1$.
	Next, for each vertex~$v\in D$ we add leaf vertices such that~$\deg(v)=\ell$.
	By~$L$ we denote the union of all these added leaf vertices.
	Finally, we set~$k'\coloneqq k+tq-q$ and~$t'\coloneqq\alpha \left[k\ell+(tq-q)(\ell+1)\right]$.

	\paragraph{Degeneracy:} Next, we show that~$d(H)=q+2$.
	Let~$F$ be any subgraph of~$H$.
	We have to show that~$F$ contains a vertex with degree at most~$d(H)=q+2$.
	Clearly, if~$F$ contains a leaf vertex of~$L$, then~$F$ contains a vertex of degree~$1$.
	Thus, in the following we can assume that~$F\cap L=\emptyset$.
	If~$F\subseteq J$ then~$F$ is edgeless.
	Hence, let~$F^D\coloneqq V(F)\cap D\ne\emptyset$ and let~$F^D_x=F^D\cap V_x\ne\emptyset$ for some~$x\in[t]^q$.
	Since by assumption~$\mathcal{I}_x$ is~$2$-degenerate, there exists a vertex~$v\in F^D_x$ such that~$|N(v)\cap F^D_x|\le 2$. 
	By construction~$v$ has only neighbors in~$V_x$ and exactly~$q$ neighbors in~$J$ an no neighbors in~$V_y$ for some~$y\ne x$.
	Thus,~$v$ has degree at most~$q+2$.

	\paragraph{Correctness:} In the following, we prove that there exists an independent set of size exactly~$k$ for some instance~$\mathcal{I}_x$ with~$x\in [t]^q$ if and only if there exists a vertex set~$S$ of size exactly~$k'$ in~$H$ such that~$\val(S) \ge t'$.

	Suppose that instance~$\mathcal{I}_x$ has an independent set~$I$ of size exactly~$k$ for some~$x\in[t]^q$.
	By~$J^*\coloneqq J\setminus\bigcup_{i\in[q]}\{w^i_{x_i}\}$ we denote the non-neighbors of~$V_x$ in~$J$.
	Note that~$|J^*|=tq-q$.
	We show that~$S\coloneqq I\cup J^*$ is a solution of~$(H,k',t)$.
	Clearly,~$|S|=k+tq-q=k'$.
	Since no vertex of~$I$ is connected with any vertex in~$J^*$,~$I$ is an independent set, and since~$J^*$ is an independent set, we conclude that~$I\cup J^*$ is an independent set. Thus,~$E_H(S)=\emptyset$.
	Furthermore, since each vertex in~$I$ has degree~$\ell$, and since each vertex in~$J^*$ has degree~$\ell+1$, we conclude that~$m_H(S, V(H) \setminus S)=k\ell+(tq-q)(\ell+1)$.
	Thus,~$\val(S)=t'$ and hence~$(H,k',t')$ is a yes-instance of \probBorMax{}.

	Conversely, suppose that~$(H,k',t')$ has a solution~$S\subseteq V(H)$ of size exactly~$k'$ with~$\val(S)\ge t'$. 
	First, we show that~$S$ cannot contain any leaf-vertex in~$L$.
	Assume towards a contradiction that~$S\cap L\ne\emptyset$. 
	Let~$v\in S\cap L$ and let~$S'\coloneqq S\setminus\{v\}$.
	According to \Cref{lem-alpha-inner-or-out},~$\val(S)$ is maximal if~$S'$ is an independent set and~$E(v,S')=\emptyset$.
	Since the maximum degree in~$H$ is~$\ell+1$ and since~$\deg(v)=1$, we obtain that~$\val(S)\le\alpha\left[1+(k'-1)(\ell+1)\right]$.
	Hence,
	\begin{align*}
	t-\val(S)&\ge\alpha k\ell+\alpha (tq-q)(\ell+1)-\left[\alpha+\alpha k\ell+\alpha k +\alpha (tq-q)(\ell+1)-\alpha(\ell+1)\right]\\
	&=\alpha(\ell-k)>0
	\end{align*}
	since~$\ell>k$.
	Hence, in the following we can assume that~$S\cap L=\emptyset$.
	Let~$J^*\coloneqq J\cap S$,~$|J^*|=z$, and~$D^*\coloneqq S\setminus J^*\subseteq D$.
	In the following, we show that~$z=(t-1)q$ and that there exists an~$x\in[t]^q$ such that~$N(V_x)\cap J^*=\emptyset$.
	Therefore, we consider the cases that~$z<tq-q$ and that~$z>tq-q$.
	In both cases we verify that~$\val(S)<t'$ for each solution with exactly~$z$ vertices in~$J$.

	\paragraph{Case~$1$:~$z\le tq-q-1$.} 
	Let~$p\coloneqq tq-q-z$. 
	Note that~$p\in[tq-q]$.
	Recall that~$\deg(v)=\ell$ for each vertex~$v\in D$.
	Thus, by \Cref{lem-alpha-inner-or-out},~$\val(S)$ is maximized if~$m_H(D^*)+m_H(J^*,D^*)$ is minimized. 
	Since~$|J^*|<tq-q$, it is possible that no vertex of~$D^*$ is adjacent to any vertex in~$J^*$.
	Thus,~$\val(S)$ is maximal if~$S$ is an independent set with exactly~$tq-q-p$ vertices in~$J$.
	Hence,~$\val(S)\le \alpha\left[(k+p)\ell+(tq-q-p)(\ell+1)\right]$.
	Now, we obtain that 

	\begin{align*}
	t-\val(S)&\ge \alpha k\ell+\alpha(tq-q)(\ell+1)-\left[\alpha k\ell+\alpha p\ell+\alpha(tq-q)(\ell+1)-\alpha p(\ell+1)\right]\\
	&=\alpha p >0.
	\end{align*}
	Thus,~$t'-\val(S)>0$, a contradiction.

	\paragraph{Case~$2$:~$z\ge tq-q+1$.}
	Let~$p\coloneqq z-tq+q$. 
	In other words~$p\in[q]$.
	By the pigeonhole principle there exist at least~$p$ indices~$i\in[q]$ such that~$w^i_j\in S$ for each~$j\in[t]$. 
	Recall that by construction, each vertex~$v\in D$ has exactly one neighbor in the set~$\{w^i_j, j\in[t]\}$.
	Since~$|D^*|=k-p$ we conclude that~$m_H(J^*,D^*)\ge (k-p)p$.
	Recall that~$\deg(v)=\ell$ for each vertex~$v\in D$.
	Thus, by \Cref{lem-alpha-inner-or-out},~$\val(S)$ is maximized if~$m_H(D^*)+m_H(J^*,D^*)$ is minimized. 
	Hence,~$\val(S)$ is maximal if~$m_H(J^*,D^*)=(k-p)p$ and~$D^*$ is an independent set.
	Since~$H[J]$ is an independent set we obtain that  
	\begin{align*}
	\val(S)&\le (1-\alpha)(k-p)p+\alpha\left[(k-p)(\ell-p)+(t-1)q(\ell+1)+p(\ell+1-k+p)\right].
	\end{align*}
	
	Now, we obtain that 
	\begin{align*}
	t-\val(S)&\ge \alpha k\ell +\alpha(tq-q)(\ell+1)-kp\\
	&\hspace{1cm} +p^2+3\alpha kp-3\alpha p^2-\alpha k\ell-\alpha(tq-q)(\ell+1)-\alpha p\\
	&=-kp+p^2+3\alpha kp-3\alpha p^2-\alpha p\\
	&= (3\alpha -1)kp + (1-3\alpha)p^2-\alpha p.
	\end{align*}

	From~$p\ge 0$ and~$3\alpha-1>0$ for~$\alpha\in(1/3,1]$, we obtain that

	\begin{align*}
	t-\val(S)&\ge (3\alpha-1)k +(1-3\alpha)p^2-\alpha p = (3\alpha-1)k+ c'(\alpha,p)
	\end{align*}

	where~$c'(\alpha,p)<0$ is a constant only depending on the parameters~$\alpha$, and~$p$. 
	We set~$c(\alpha)\coloneqq \max_{p\in[q]}c'(\alpha,p)$ which is a constant smaller than~$0$ only depending on~$\alpha$.
	In other words,~$t-\val(S)\ge (3\alpha-1)k+c(\alpha)>0$ since by assumption~$k>|c(\alpha)|/(3\alpha-1)$.
	A contradiction to the fact that~$\val(S)\ge t'$.

	Hence,~$|J^*|=tq-q$ and thus~$|D^*|=k$.
	According to \Cref{lem-alpha-inner-or-out},~$\val(S)$ is maximal if~$m_H(D^*)+m_H(J^*,D^*)$ is minimal.
	Observe that if~$E_H(J^*,D^*)=\emptyset$ and~$D^*$ is an independent set, then~$\val(S)=t'$.
	Otherwise, if~$E_H(J^*,D^*)\ne\emptyset$ or if~$D^*$ is no independent set, then~$\val(S)<t'$.
	Thus, $E_H(J^*,D^*)=\emptyset$ and~$D^*$ is an independent set.
	Now, if there exist two vertices~$u,v\in D^*$ such that~$u\in V_x$ and~$v\in V_{y}$ with~$x\ne y$ for~$x,y\in [t]^q$, then~$(N(u)\cup N(v))\cap J\ge q+1$.
	Since~$|J^*|=tq-q$ this implies that~$E_H(J^*,D^*)\ne\emptyset$, a contradiction.
	Hence,~$D^*\subseteq V_x$ for some~$x\in[t]^q$.
	Furthermore,~$J^*=J\setminus N(D^*)$.
	Thus, the instance~$x$ contains an independent set of size at least~$k$.

	Hence, we have a weak-$q$-composition from \textsc{Independent Set} to \probBorMax{} in~$(d+2)$-closed graphs.
	Now, the proposition follows by \Cref{lemma:is:noq}.
\end{proof}

\section{Parameterization by  $h$-Index and Vertex Cover Number}

\label{sec:vc}
To complete the picture of the parameterized complexity landscape, we consider two parameters that are larger than the degeneracy of~$G$: the $h$-index of~$G$ and the vertex cover number of~$G$. 

Our results in this section are based on two data reduction rules.
The first rule discards (according to the Exclusion Rule (\cref{rr:general:how-to-exclude})) vertices with small contribution when there are sufficiently many vertices with high contribution.
The second rule adds vertices with very large contribution to a solution (according to the Inclusion Rule (\cref{rr:general:how-to-include})) assuming there are only few vertices with large contribution.
Below, we will specify when the contribution is small, or large.

\begin{definition}\label{def:high-contribution-vertices}
	Let~$\mathcal{I}$ be an instance of \probAno{} and let~$x \in \mathds{N}$.
	Then~$V_x \coloneqq \{v \in V(G) \mid \degCounter(v) \ge x\}$.
\end{definition}

This definition helps us to specify when the contribution is small, or large.

\begin{lemma}\label{lem:alpha-k-interval:no-small-contribution}
	Let~$\mathcal{I}$ be a yes-instance of \probAnoMax{}, let~$x \in \mathds{N}$ with~$|V_x| \ge k$, and let~$v \in V(G)$ be a vertex with~$\alpha\cdot\degCounter(v) < \alpha x - |(1 - 3\alpha)k|$.
	Then, there is a solution~$S$ with~$v \notin S$.
\end{lemma}

\begin{proof}
	Assume towards a contradiction that~$v$ is contained in each solution~$S$.
	Observe that~$v \notin V_x$.
	Since~$|S| = k$ it follows that there is a vertex~$u \in V_x \setminus S$.
	We claim that~$S' \coloneqq  (S \setminus \{v\}) \cup \{u\}$ is also a solution.
	This follows from \Cref{lem:general:condition-strictly-better,lemma:replace}.
\end{proof}

\begin{lemma}\label{lem:alpha-k-interval:all-high-contribution}
	Let~$\mathcal{I}$ be a yes-instance of \probAnoMax{}, let~$x \in \mathds{N}$ with~$|V_x| \le k$, and let~$v \in V(G)$ be a vertex with~$\alpha\degCounter(v) \ge \alpha x + |(1 - 3\alpha)k|$.
	Then, there is a solution~$S$ with~$v \in S$.
\end{lemma}

\begin{proof}
	Assume towards a contradiction that~$v$ is not contained in any solution~$S$.
	Since~$|S| = k$, $v \in V_x \setminus S$, and~$|V_x| \le k$ it follows that there is a vertex~$u \in S \setminus V_x$.
	We claim that~$S' \coloneqq  (S \setminus \{u\}) \cup \{v\}$ is also a solution.
	This follows from \Cref{lem:general:condition-strictly-better,lemma:replace}.
\end{proof}

\subsection{Parameterization by $h$-Index}

We start with the maximization variant and the~$h$-index.
As \probIntMax{} does not admit a polynomial kernel with respect to~$k$ even if~$\Delta$ is constant (see \Cref{thm:maxdeg:no-poly-kernel-const-delta}), the same holds for the~$h$-index.
We show that in contrast, the \borderFocusedS{} case admits a polynomial kernel.

\begin{proposition}
	\label{prop:h-index:poly-kernel}
	\probBorMax{} admits a kernel of size $\Oh(\alpha^{-2} (h^2 k^2 + k^4))$.
\end{proposition}

To show this result we make use of the two rules discarding (\Cref{lem:alpha-k-interval:no-small-contribution}) and adding (\Cref{lem:alpha-k-interval:all-high-contribution}) vertices with small or large contribution, respectively.

\begin{lemma}
	\label{lem:h-index:technical-kernel}
	Let~$\mathcal{I}=(G,k,t)$ be an instance of \probMax{}.
	\begin{enumerate}
		\item If~$\alpha > 0$ and there are at least~$k$ vertices with degree at least~$h + 1 + |(1 - 3\alpha)k\cdot \alpha^{-1}|$, then an equivalent instance of size~$\Oh(h^2 + \alpha^{-1} hk^2)$ can be computed in polynomial time.
		\item If~$\alpha > 1/3$ and there are less than~$k$ vertices with degree at least~$h + 1 + |(1 - 3\alpha)k\cdot \alpha^{-1}|$, then an equivalent instance of size~$\Oh(\alpha^{-2}(h^2k^2 + k^4))$ can be computed in polynomial time.
	\end{enumerate}
\end{lemma}

\begin{proof}
	First, we transform~$\mathcal{I}$ into an equivalent instance~$(G,\emptyset,\counter,k,t)$ of \probAnoMax{} where~$\counter(v) = 0$ for all~$v \in V(G)$.
	Then~$V_{h+1}$ is the set of vertices of degree greater than~$h$.
	By the definition of $h$-index, we have $|V_{h + 1}| \le h$.
	We set~$Y \coloneqq  V(G) \setminus V_{h+1}$ to be the vertices with degree at most~$h$.
	Let $x \coloneqq  h + 1 + |(1 - 3\alpha)k\cdot \alpha^{-1}|>h$ (recall~$\alpha > 0$).
	Then, we have $V_x \subseteq V_{h+1}$.
	We will assume that $|V(G)| > k$.
	
	(1): $|V_x| \ge k$ and~$\alpha > 0$. 
	For each vertex~$v \in Y$, we have $\deg(v) = \degCounter(v) \le h < x - |(1 - 3\alpha)k\cdot \alpha^{-1}|$.
	Hence, by \Cref{lem:alpha-k-interval:no-small-contribution}, there is a solution not containing~$v$.
	We thus can apply the Exclusion Rule (\Cref{rr:general:how-to-exclude}) on~$v$.
	Since this application does not change~$\degCounter(u)$ for any~$u \in Y$, we can apply the Exclusion Rule (\Cref{rr:general:how-to-exclude}) on all vertices in~$Y$ to obtain a graph with~$h$ vertices (the vertices in~$V_{h+1}$).

	Removing annotations by \Cref{lemma:general:remove-ano} results in an instance whose size is bounded in terms of~$\Delta$ and $\Gamma$.
	Thus, we need to bound these two parameters since every vertex not in $V_{h + 1}$ has been deleted.
	We clearly have $\Delta \le |V(G)| = h$.
	To bound $\Gamma$, we apply the following procedure:
	We apply \Cref{rr:general:no-zero-counter} exhaustively throughout.
	We then end up with a vertex $v$ with $\counter(v) = 0$.
	If there exists a vertex $u \in V(G)$ with $\counter(u) > \deg(v) + |(1 - 3 \alpha)k \cdot \alpha^{-1}|$, then we can apply the Exclusion Rule (\Cref{rr:general:how-to-exclude}) on $u$ because $u$ is strictly better than $v$ by \Cref{lem:general:condition-strictly-better}.
	After this procedure, we may assume that $\Gamma \le h + |\alpha^{-1} - 3| k$.
	By \Cref{lemma:general:remove-ano}, we obtain an equivalent instance of size $\Oh((\Delta + \Gamma + \alpha^{-1}) \cdot |V(G)|) = \Oh(h^2 + \alpha^{-1} h k)$.
	
	(2): $|V_x| < k$ and $\alpha > 1/3$. 
	Consider the set~$V_{x+|(1 - 3\alpha)k\cdot \alpha^{-1}|}$ of vertices with degree more than~$h + 2|(1 - 3\alpha)k\cdot \alpha^{-1}|$.
	By \Cref{lem:alpha-k-interval:all-high-contribution} for each~$v \in V_{x+|(1 - 3\alpha)k\cdot \alpha^{-1}|}$ there is a solution containing~$v$.
	Thus, we can apply the Inclusion Rule (\Cref{rr:general:how-to-include}) on all vertices in~$V_{x+|(1 - 3\alpha)k\cdot \alpha^{-1}|}$ to obtain an instance with $\Delta_{\overline{T}} \le x+|(1 - 3\alpha)k\cdot \alpha^{-1}|$ (recall that $\Delta_{\overline{T}}$ is the maximum degree over vertices not in~$T$). 
	The exhaustive application of \Cref{rr:delta:better} results in a graph with at most $$\Delta_{\overline{T}}k + 1 \le xk+|(1 - 3\alpha)k\cdot \alpha^{-1}|k + 1 = (h + 1)k + \alpha^{-1}|2/\alpha - 6|k^2 + 1 = \Oh(hk + \alpha^{-1}k^2)$$ vertices.
	Since we are dealing with the degrading case, we can use \Cref{prop:general:remove-ano-n-bounded} to obtain an instance for \probMax{} of size 
	$$\Oh(|V(G)|^2+\alpha^{-1}|V(G)|k^2)\subseteq\Oh((hk + \alpha^{-1}k^2)^2+\alpha^{-1}(hk+\alpha^{-1}k^2)k^2) = \Oh(\alpha^{-2}(h^2k^2 + k^4)).$$
	Thus, the statement follows.
\end{proof}

\Cref{lem:h-index:technical-kernel} implies \Cref{prop:h-index:poly-kernel} and thereby the existence of polynomial kernel for~$\alpha>1/3$: apply (1) if $k$ vertices have degree at least $h + 1 + |(1 - 3 \alpha)\cdot \alpha^{-1}|$ and (2) otherwise.
It is unlikely that \Cref{lem:h-index:technical-kernel} (2) can be extended to cover the case $\alpha \in (0, 1/3)$.
This would imply that \probIntMax{} admits a polynomial kernel with respect to $k + h$, contradicting \Cref{thm:maxdeg:no-poly-kernel-const-delta} (which states that \probIntMax{} does not admit a polynomial kernel with respect to~$k$ for constant~$\Delta$).

We complement this with showing fixed-parameter tractability for~$k+h$.
\begin{proposition}
	\label{prop:hindex:fpt-internal-max}
	\probIntMax{} is fixed-parameter tractable with respect to~$k+h$.
\end{proposition}

{
\begin{proof}
  We first transform~$\mathcal{I}$ into an equivalent
  instance~$(G,\emptyset,\counter,k,t)$ of \probAnoMax{}
  where~$\counter(v) = 0$ for all~$v \in V(G)$. To make the description of the algorithm easier, we redefine $\val$ by
$$\val_G(S) \coloneqq \alpha (m(S, V(G) \setminus S)) + \counter(S) + (1 - \alpha) m(S),$$
that is, we do not multiply the counter by~$\alpha$ but instead add correct multipliers when updating the counter.

  Now, let~$V_{>h}$ be the set
  of vertices of degree at least~$h+1$.  For each subset~$T$ of size at
  most~$k$ of~$V_{>h}$, branch into the case that~$T=S\cap V_{>h}$. In each case, we may now apply the Exclusion Rule (\Cref{rr:general:how-to-exclude}) to remove the vertices of~$V_{>h}\setminus T$. Now the contribution of~$T$ depends on which vertices of~$V(G)\setminus V_{>h}$ are contained in~$S$. This can be incorporated into the counters of~$V(G)\setminus V_{>h}$ as follows. Pick any vertex~$u\in T$.  Then, decrease~$t$ by~$\counter(u) + (1-\alpha) |N(u)\setminus T| + \alpha |N(u)\setminus T|$. Now, for each vertex~$v\in N(u)\setminus T=N(u)\setminus V_{>h}$, add~$(1-\alpha)-\alpha=1-2\alpha$ to ~$\counter(v)$. The correctness of this step can be seen as follows: if~$v$ is not contained in~$S$, then the contribution of~$uv$ is~$\alpha$ and this contribution is already recorded in the decrease of~$t$. However, if we also add~$v$ to~$u$, then the contribution of~$uv$ is~$1-\alpha$, so this gives a value of~$1-\alpha$ for the internal edge $uv$ minus $\alpha$ for the fact that~$uv$ is no longer an outgoing edge. Finally, remove~$u$ from~$G$.

  After this removal of vertices in~$T$ has been applied exhaustively,
  the remaining graph has only vertices of degree at most~$h$. Our aim
  is to find in this graph a vertex set~$S'$ of size~$k-|T|$ that
  maximizes~$\val(S')$. Now this problem can be solved in~$f(h,k)\cdot n^{\Oh(1)}$~time since~$\val$ fulfills a property of fixed-cardinality optimization problems called \emph{component linear} by Komusiewicz and Sorge~\cite{KS15}: First,~$\val(S\cup T)\le \val(S)+\val(T)$ because an internal edge counts twice as much as an outgoing edge. Second,~$\val(S\cup T)\ge \val(S)+\val(T)$ if there are no edges between~$S$ and~$T$ in~$G$. 

  Altogether, the running time is~$h^k\cdot n^{\Oh(1)}\cdot f(h,k)$ since we create~$h^k$ many cases in the branching on~$V_{>h}$.
\end{proof}
}

\paragraph{Minimization variant.}

Note that \probBorMin{} has a polynomial kernel with respect to~$d+k$ (see \Cref{thm-minimization-kernel-d+k}) and, thus, also with respect to~$h+k$.
As \probIntMin{} does not admit a polynomial kernel with respect to~$k$ even if~$\Delta$ is constant (see \Cref{thm:maxdeg:no-poly-kernel-const-delta}), the same holds for the~$h$-index.

\subsection{Parameterization by Vertex Cover Number}

We have shown that \probMax{} admits a polynomial kernel with respect to $h + k$ for $\alpha > 1/3$.
For the larger parameter vertex cover number~$\mathsf{vc}$, we achieve a polynomial kernel for all~$\alpha > 0$. 

\begin{proposition}
	\label{prop:vertex-cover:poly-kernel-max}
	If~$\alpha \ne 0$, then \probMax{} admits a kernel of size~$\Oh(\mathsf{vc} (\mathsf{vc} + \alpha^{-1} k)^2)$.
\end{proposition}

{
\begin{proof}
	We follow the proof of \cref{lem:h-index:technical-kernel}.
	As $\mathsf{vc} \ge h$, we only need to extend the statement of \cref{lem:h-index:technical-kernel} (2) concerning the case $\alpha\in(0,1/3)$.
	Let~$(G,k,t)$ be an instance of \probMax{}.
	First, we transform~$(G,k,t)$ into an equivalent instance~$(G,\emptyset,\counter,k,t)$ of \probAnoMax{} where~$\counter(v) = 0$ for all~$v \in V(G)$.
	Then, let~$X$ be a vertex cover of size~$\mathsf{vc}$.
	We set~$I \coloneqq V(G) \setminus X$ to be the independent set.
	Note that each vertex in~$I$ has degree at most~$\mathsf{vc}$.
	Moreover, we set~$x \coloneqq\mathsf{vc} + |(1 - 3\alpha)k\cdot \alpha^{-1}|>\mathsf{vc}$ since~$\alpha > 0$.
	Thus, $V_x \subseteq X$.

	
	\paragraph{Case 1: $|V_x| \ge k$.} 
	This case follows from \cref{lem:h-index:technical-kernel} (1).

	\paragraph{Case 2: $|V_x| < k$.} 
	
	Consider the set~$V_{x+|(1 - 3\alpha)k\cdot \alpha^{-1}|}$ of vertices with degree at least~$\mathsf{vc} + 2|(1 - 3\alpha)k\cdot \alpha^{-1}|$.
	By \Cref{lem:alpha-k-interval:all-high-contribution}, for each~$v \in V_{x+|(1 - 3\alpha)k\cdot \alpha^{-1}|}$ there is a solution containing~$v$.
	Thus, we can apply the Inclusion Rule (\Cref{rr:general:how-to-include}) on every vertex in~$V_{x+|(1 - 3\alpha)k\cdot \alpha^{-1}|}$ including it into $T$.
	(Recall that~$T$ is the set of vertices fixed in the solution.)
	We then have~$\Delta_{\overline{T}} \le x+|(1 - 3\alpha)k\cdot \alpha^{-1}|$.
	Thus, there are at most~$\mathsf{vc} \cdot \Delta_{\overline{T}} \in \Oh(\mathsf{vc}^2 + \alpha^{-1}\mathsf{vc}\cdot k)$ vertices in~$I$ that have at least one neighbor in~$X \setminus T$.
	Denote these vertices by~$I_{\overline{T}}$.
	The remaining vertices in~$I \setminus I_{\overline{T}}$ have neighbors only in~$T$. 
	Hence, their contribution is fixed and we can get rid of all but the~$k$ vertices in~$I \setminus I_{\overline{T}}$ with highest contribution using the Exclusion Rule (\Cref{rr:general:how-to-exclude}).
	Thus, we are left with the vertices in~$X$, in~$I_{\overline{T}}$, and the $k$ vertices with highest contribution of~$I \setminus I_{\overline{T}}$. 
	These are $\Oh(\mathsf{vc}^2 + \alpha^{-1} \mathsf{vc}\cdot k)$ many vertices.
	We remove the annotation using \Cref{lemma:general:remove-ano}: we obtain an instance for \probMax{} of size~$\Oh((\Delta_{\overline{T}} + \Gamma + \alpha^{-1}) |V(G)| + \alpha^{-1} k |T|) = \Oh(\mathsf{vc} (\mathsf{vc} + \alpha^{-1} k)^2)$.
\end{proof}
}

For~$\alpha = 0$, \probMax{} corresponds to \textsc{Densest~$k$-Subgraph} and  \textsc{Clique} is one of its special cases ($t = \binom{k}{2}$).
Since \textsc{Clique} does not admit a polynomial kernel with respect to $\mathsf{vc}$~\cite{BJK14} (and any clique is of size at most~$\mathsf{vc}+1$), \textsc{Densest~$k$-Subgraph} does not admit a polynomial kernel with respect to~$\mathsf{vc}$.
However, \textsc{Densest~$k$-Subgraph} can be solved by a straightforward  algorithm in $\Oh^*(2^\mathsf{vc})$~time.
Thus, \textsc{Densest~$k$-Subgraph} admits a kernel of size~$\Oh(2^\mathsf{vc})$.

\paragraph{Minimization variant.}
Recall that \probBorMin{} has a polynomial kernel with respect to~$d+k$ (see \Cref{thm-minimization-kernel-d+k}) and thereby $\mathsf{vc}+k$.
It remains to consider \probIntMin{} parameterized by~$\mathsf{vc}+k$.

\begin{proposition}
	\label{prop:vertex-cover:poly-kernel-min}
	\probMin{} admits a kernel of size~$\Oh((\alpha^{-2} + k) (\mathsf{vc} + \alpha^{-1} \mathsf{vc} \cdot k)^2)$ for~$\alpha > 0$ and of size~$\Oh(\mathsf{vc}^2 + \mathsf{vc}\cdot k)$ for $\alpha = 0$.
\end{proposition}

\begin{proof}
	Let~$X$ be a vertex cover of size~$\mathsf{vc}$ and let~$I \coloneqq  V(G) \setminus X$ be the independent set.
	Without loss of generality we can assume that~$|I|\ge k$ since otherwise we already have a trivial~$\mathsf{vc} + k$-vertex kernel.
	If~$\alpha = 0$, we have a trivial yes-instance as~$\val(I') = 0$ for all~$I' \subseteq I$ of size~$k$.
	Thus, in the following, we assume that~$\alpha > 0$.
	Let~$(G,k,t)$ be an instance of \probMin{}.
	We transform~$(G,k,t)$ into an equivalent instance~$(G,\emptyset,\counter,k,t)$ of \probAnoMin{} where~$\counter(v) = 0$ for all~$v \in V(G)$.

	Let~$x \coloneqq \mathsf{vc} + |(1-3\alpha)k\cdot \alpha^{-1}|$.
	We first show that there is a solution that does not contain any vertex in~$V_x$:
	To this end, observe that~$\deg(v) \le \mathsf{vc}$ for each~$v \in I$.
	Hence, by \cref{lem:general:condition-strictly-better}, each vertex in~$I$ is strictly better than each vertex in~$V_x$.
	Since~$|I| \ge k$, it follows from \cref{lemma:replace} that there is a solution not containing any vertex from~$V_x$.
	Hence, we can apply the Exclusion Rule (\cref{rr:general:how-to-exclude}) on each vertex in~$V_x$.
	As a result, the remaining vertices in~$X$ form still a vertex cover and have degree less than~$\mathsf{vc} + |(1-3\alpha)k\cdot \alpha^{-1}|$.
	Thus, less than~$\mathsf{vc}(\mathsf{vc} + |(1-3\alpha)k\cdot \alpha^{-1}|)$ vertices in~$I$ have neighbors in~$X$; the remaining vertices are isolated vertices.
	As the contribution of each isolated vertex~$v$ in any solution is exactly~$\alpha \cdot \counter(v)$, we can simply sort the isolated vertices by their contribution and remove all but the~$k$ vertices with the lowest contribution.
	Thus, we end up with at most $\mathsf{vc} + \mathsf{vc}(\mathsf{vc} + |(1-3\alpha)k\cdot \alpha^{-1}|) + k = \Oh(\mathsf{vc}^2 + \alpha^{-1}\mathsf{vc}\cdot k)$~vertices.
	Using \Cref{lemma:general:remove-ano-min} to remove the annotation, we get an instance for \probMax{} of size 
	\begin{align*}
		\Oh(\alpha^{-2} (\Delta + \Gamma + k)^2 + \alpha^{-1} (\Delta + \Gamma + k) \cdot |V(G)|) = \Oh((\alpha^{-2} + k) (\mathsf{vc} + \alpha^{-1} \mathsf{vc} \cdot k)^2)
	\end{align*}
	Thus, the statement follows.
\end{proof}

\section{Conclusion}

We provided a systematic parameterized complexity analysis for \prob{} (see \cref{fig:overview}).
Although we settled the existence of polynomial kernels with respect to various parameters combined with the solution size~$k$, several open questions remain.
First, our polynomial kernels are not optimized and thus the polynomials are of high degree. 
Looking for smaller kernels is thus an obvious first task. 
Second, can our positive results for $c$-closure and degeneracy be extended to the smaller parameter weak closure \cite{FRSWW20}?
Furthermore, while we looked at parameters that are small in sparse graphs, can similar results be achieved for dense graphs as considered e.\,g. by Lochet et al.~\cite{LL0Z21}?
Moreover, it would be interesting to perform a similar complete study for related problems. A natural candidate could be the variant of \prob{} where an additional connectivity constraint is imposed on~$G[S]$; the special case of this variant where one aims to maximize the number of edges with one endpoint in~$S$ has been studied under the name \textsc{Multi-Node Hub} problem~\cite{SZ18b}.  
Finally, we believe that an experimental verification of our data reduction rules could demonstrate their practical usefulness. 

\bibliographystyle{plain}

\end{document}